\documentclass[sigconf, screen, authorversion]{acmart_modified}

\usepackage{amsmath}
\usepackage{amsfonts}
\usepackage{amssymb}
\usepackage{amsthm}
\usepackage[utf8]{inputenc}
\usepackage{microtype}
\usepackage{enumitem}
\usepackage{graphicx}
\usepackage{dsfont} 
\usepackage{booktabs}
\usepackage{float}
\usepackage{verbatim}
\usepackage{bbm} 
\usepackage{mathrsfs}
\usepackage{array}
\usepackage{booktabs, multicol, multirow}
\usepackage{cellspace}
\usepackage{makecell}
\usepackage{diagbox}
\usepackage[caption = false]{subfig}
\usepackage{flushend} 

\theoremstyle{acmdefinition}

\newcommand{\LSH}{\mathcal{H}}
\newcommand{\LSM}{\mathcal{M}}
\newcommand{\E}{\mathrm{E}}
\newcommand{\Var}{\mathrm{Var}}

\DeclareMathOperator{\poly}{poly}
\DeclareMathOperator{\sign}{sign}

\DeclareMathOperator*{\argmin}{arg\,min}

\newcommand{\q}{\mathbf{q}}
\newcommand{\x}{\mathbf{x}}
\newcommand{\y}{\mathbf{y}}
\newcommand{\z}{\mathbf{z}}
\newcommand{\vect}[1]{\mathbf{#1}}
\newcommand{\1}{\mathds{1}}

\newcommand{\real}{\mathbb{R}}

\newcommand{\norm}[1]{\lVert #1 \rVert} 
\newcommand{\cube}[1]{\{0,1\}^{#1}}
\newcommand{\K}{\mathbb{K}}

\newcommand{\BO}[1]{O({#1})}

\newcommand{\sen}[1]{\K_{\LSH}(#1)} 

\newcommand{\para}[1]{\textsf{\textbf{#1}}}

\renewcommand{\arraystretch}{1.5}

\setcopyright{none} 
\acmDOI{10.1145/3055399.3055443}
\acmISBN{978-1-4503-4528-6/17/06}
\acmConference[STOC'17]{49th Annual ACM SIGACT Symposium on the Theory of Computing}{June 2017}{Montreal, Canada}
\acmYear{2017}
\copyrightyear{2017}
\acmPrice{15.00}

\usepackage{fancyhdr}
\begin{document}

\title{Set Similarity Search Beyond MinHash}

\titlenote{The research leading to these results has received
funding from the European Research Council under the European Union’s 7th
Framework Programme (FP7/2007-2013) / ERC grant agreement no.~614331.}

\author{Tobias Christiani}
\affiliation{%
  \institution{IT University of Copenhagen}
  \city{Copenhagen}
  \country{Denmark}
}
\email{tobc@itu.dk}

\author{Rasmus Pagh}
\affiliation{%
  \institution{IT University of Copenhagen}
  \city{Copenhagen}
  \country{Denmark}
}
\email{pagh@itu.dk}

\begin{abstract}
We consider the problem of approximate set similarity search under Braun-Blanquet similarity $B(\x, \y) = |\x \cap \y| / \max(|\x|, |\y|)$.
The $(b_1, b_2)$-approximate Braun-Blanquet similarity search problem is to preprocess a collection of sets $P$ such that, given a query set $\q$, 
if there exists $\x \in P$ with $B(\q, \x) \geq b_1$, then we can efficiently return $\x' \in P$ with $B(\q, \x') > b_2$.

We present a simple data structure that solves this problem with space usage $O(n^{1+\rho}\log n + \sum_{\x \in P}|\x|)$ and query time $O(|\q|n^{\rho} \log n)$ where $n = |P|$ and $\rho = \log(1/b_1)/\log(1/b_2)$.
Making use of existing lower bounds for locality-sensitive hashing by O'Donnell et al. (TOCT 2014) we show that this value of $\rho$ is tight across the parameter space, i.e., for every choice of constants $0 < b_2 < b_1 < 1$. 

In the case where all sets have the same size our solution strictly improves upon the value of $\rho$ that can be obtained through the use of state-of-the-art data-independent techniques in the Indyk-Motwani locality-sensitive hashing framework (STOC 1998) such as Broder's MinHash (CCS 1997) for Jaccard similarity and Andoni et al.'s cross-polytope LSH (NIPS 2015) for cosine similarity.
Surprisingly, even though our solution is data-independent, for a large part of the parameter space we outperform the currently best data-\emph{dependent} method by Andoni and Razenshteyn (STOC 2015).
\end{abstract}

\maketitle
\thispagestyle{fancy}
\section{Introduction}
In this paper we consider the approximate set similarity problem or, 
equivalently, the problem of approximate Hamming near neighbor search in sparse vectors. 
Data that can be represented as sparse vectors is ubiquitous --- a typical example is the representation of text documents as \emph{term vectors}, 
where non-zero vector entries correspond to occurrences of words (or shingles).
In order to perform identification of near-identical text documents in web-scale collections, 
Broder et al.~\cite{Bro97b,Bro97a} designed and implemented \emph{MinHash} (a.k.a.~min-wise hashing), 
now understood as a locality-sensitive hash function~\cite{har-peled2012}.
This allowed approximate answers to similarity queries to be computed much faster than by other methods, 
and in particular made it possible to cluster the web pages of the AltaVista search engine (for the purpose of eliminating near-duplicate search results).
Almost two decades after it was first described, 
MinHash remains one of the most widely used locality-sensitive hashing methods as witnessed by thousands of citations of~\cite{Bro97b,Bro97a}. 

A \emph{similarity measure} maps a pair of vectors to a similarity score in $[0;1]$.
It will often be convenient to interpret a vector $\x\in\{0,1\}^d$ as the set $\{ i \; | \; \x_i=1\}$.
With this convention the \emph{Jaccard similarity} of two vectors can be expressed as $J(\x,\y) = |\x \cap \y|/|\x \cup \y|$.
In \emph{approximate similarity search} we are interested the problem of searching a data set $P\subseteq \{0,1\}^d$ for a vector 
of similarity at least~$j_1$ with a query vector $\q \in \{0,1\}^d$, but allow the search algorithm to return a vector of similarity  $j_2 < j_1$. 
To simplify the exposition we will assume throughout the introduction that all vectors are $t$-sparse, i.e., have the same Hamming weight $t$.

Recent theoretical advances in data structures for approximate \emph{near neighbor} search in Hamming space~\cite{andoni2015optimal} make it possible to beat the asymptotic performance of MinHash-based Jaccard similarity search (using the LSH framework of~\cite{har-peled2012}) in cases where the similarity threshold $j_2$ is not too small.
However, numerical computations suggest that MinHash is always better when $j_2<1/45$. 

In this paper we address the problem: Can similarity search using MinHash be improved \emph{in general}?
We give an affirmative answer in the case where all sets have the same size $t$ by introducing \mbox{\sc Chosen Path}: a simple data-independent search method that strictly improves MinHash, 
and is always better than the data-dependent method of~\cite{andoni2015optimal} when $j_2<1/9$.
Similar to data-independent locality-sensitive filtering (LSF) methods~\cite{becker2016, laarhoven2015, christiani2017framework} our method works 
by mapping each data (or query) vector to a set of keys that must be stored (or looked up).
The name {\sc Chosen Path} stems from the way the mapping is constructed: As paths in a layered random graph where the vertices at each layer is identified with the set $\{1,\dots,d\}$ of dimensions, 
and where a vector $\x$ is only allowed to choose paths that stick to non-zero components $\x_i$.
This is illustrated in Figure~\ref{fig:branchingprocess}.
\begin{figure}[b]   
	\centering
	\includegraphics[width=0.4\textwidth]{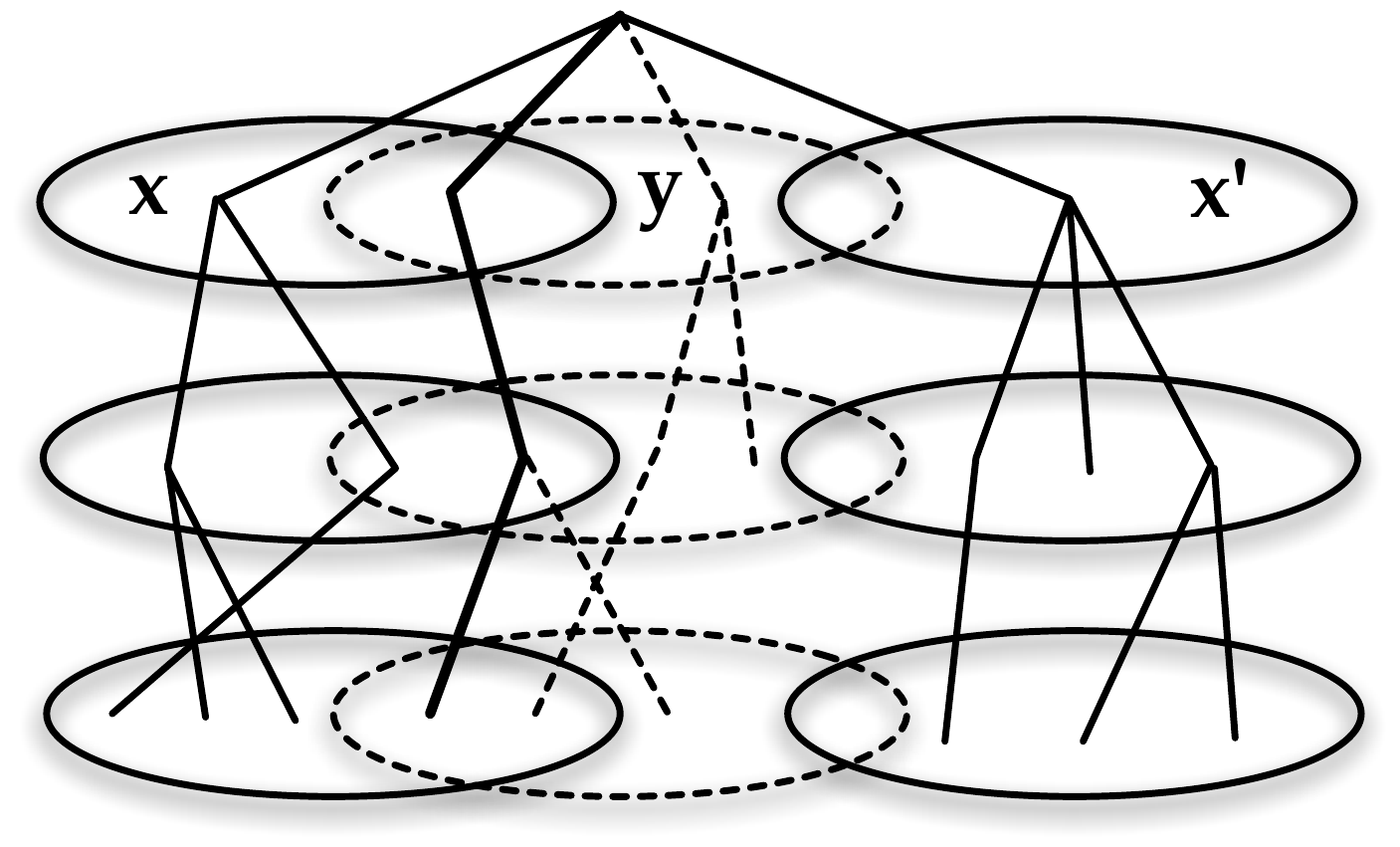}
	\caption{{\sc Chosen Path} uses a branching process to associate each vector $\x\in \cube{d}$ with a set $M_k(\x) \subseteq \{1,\dots,d\}^k$ of paths of lengtk~$k$ (in the picture $k=3$). 
	The paths associated with $\x$ are limited to indices in the set $\{ i \; | \; \x_i=1\}$, represented by an ellipsoid at each level in the illustration. 
	In the example the set sizes are: $|M_3(\x)| = 4$ and $|M_3(\y)| = |M_3(\x')| = 3$. 
	Parameters are chosen such that a query~$\y$ that is similar to $\x \in P$ is likely to have a common path in $\x \cap \y$ (shown as a bold line), 
	whereas it shares few paths in expectation with vectors such as $\x'$ that are not similar.}
	\label{fig:branchingprocess}
\end{figure}
\subsection{Related Work}
High-dimensional approximate similarity search methods can be characterized in terms of their \emph{$\rho$-value} which is the exponent for which queries can be answered in time $\BO{dn^\rho}$, 
where $n$ is the size of the set $P$ and $d$ denotes the dimensionality of the space.
Here we focus on the ``balanced'' case where we aim for space $\BO{n^{1+\rho} + dn}$, 
but note that there now exist techniques for obtaining other trade-offs between query time and space overhead~\cite{andoni2017optimal,christiani2017framework}.

\para{Locality-Sensitive Hashing Methods.}
%
We begin by describing results for Hamming space, which is a special case of similarity search on the unit sphere (many of the results cited apply to the more general case).
In Hamming space the focus has traditionally been on the $\rho$-value that can be obtained for solutions to the \emph{$(r, cr)$-approximate near neighbor} problem: 
Preprocess a set of points $P \subseteq \cube{d}$ such that, given a query point $\q$, if there exists $\x \in P$ with $\norm{\x - \q}_{1} \leq r$, then return $\x' \in P$ with $\norm{\x' - \q}_1 < cr$.
In the literature this problem is often presented as the $c$-approximate near neighbor problem where bounds for the $\rho$-value are stated in terms of $c$ and, 
in the case of upper bounds, hold for every choice of $r$, while lower bounds may only hold for specific choices of $r$. 

O'Donnell et al.~\cite{odonnell2014optimal} have shown that the value $\rho=1/c$ for $c$-approximate near neighbor search in Hamming space, 
obtained in the seminal work of Indyk and Motwani~\cite{indyk1998}, is the best possible in terms of $c$ for schemes based on Locality-Sensitive Hashing (LSH). 
However, the lower bound only applies when the distances of interest, $r$ and~$cr$, are relatively small compared to $d$, and better upper bounds are known for large distances.
Notably, other LSH schemes for angular distance on the unit sphere such as cross-polytope LSH~\cite{andoni2015practical} give lower $\rho$-values for large distances.
Extensions of the lower bound of~\cite{odonnell2014optimal} to cover more of the parameter space were recently given in~\cite{andoni2017optimal,christiani2017framework}.
Until recently the best $\rho$-value known in terms of $c$ was $1/c$, 
but in a sequence of papers Andoni et al.~\cite{andoni2014beyond,andoni2015optimal} have shown how to use \emph{data-dependent} LSH techniques to achieve $\rho = 1/(2c-1) + o_n(1)$, 
bypassing the lower bound framework of~\cite{odonnell2014optimal} which assumes the LSH to be independent of data.
\medskip

\para{Set Similarity Search.}
There exists a large number of different measures of set similarity with various applications for which it would be desirable to have efficient approximate similarity search algorithms~\cite{choi2010survey}.  
Given a measure of similarity $S$ assume that we have access to a family $\LSH$ of locality-sensitive hash functions (defined in Section \ref{sec:preliminaries})
such that for $h \sim \LSH$ it holds for every pair of sets $\x, \y$ that $\Pr[h(\x) = h(\y)] = S(\x, \y)$. 
Then we can use the the LSH framework to construct a solution for the $(s_1, s_2)$-approximate similarity search problem under $S$ with exponent $\rho = \log(1/s_1)/\log(1/s_2)$.
With respect to the existence of such families Charikar~\cite{charikar2002} showed that if the similarity measure $S$ admits an LSH with the above properties, then $1-S$ must be a metric.
Recently, Chierichetti and Kumar \cite{chierichetti2015} showed that,
given a similarity $S$ that admits an LSH with the above properties, 
the transformed similarity $f(S)$ will continue to admit an LSH if $f(\cdot)$ is a probability generating function.
The existence of an LSH that admits a similarity measure $S$ will therefore give rise to the existence of solutions to the approximate similarity search problem for the much larger class of similarities $f(S)$.
However, this still leaves open the problem of finding efficient explicit constructions, and as it turns out,
the LSH property $\Pr[h(\x) = h(\y)] = S(\x, \y)$, while intuitively appealing and useful for similarity estimation,
does not necessarily imply that the LSH is optimal for solving the approximate search problem for the measure $S$. 
The problem of finding tight upper and lower bounds on the $\rho$-value that can be obtained through the LSH framework for data-independent $(s_1, s_2)$-approximate similarity search 
across the entire parameter space $(s_1, s_2)$ remains open for two of the most common measures of set similarity: 
Jaccard similarity $J(\x, \y) = |\x \cap \y| / |\x \cup \y|$ and cosine similarity $C(\x, \y) = |\x \cap \y|/\sqrt{|\x||\y|}$.

A random function from the MinHash family $\LSH_{\text{minhash}}$ hashes a set $\x \subset \{1, \dots, d \}$
to the first element of $\x$ in a random permutation of the set $\{1, \dots, d \}$.
For $h \sim \LSH_{\text{minhash}}$ we have that $\Pr[h(\x) = h(\y)] = J(\x, \y)$, yielding an LSH solution to the approximate Jaccard similarity search problem.
For cosine similarity the SimHash family $\LSH_{\text{simhash}}$, introduced by Charikar~\cite{charikar2002}, 
works by sampling a random hyperplane in $\real^{d}$ that passes through the origin and hashing $\x$ according to what side of the hyperplane it lies on.
For $h \sim \LSH_{\text{simhash}}$ we have that $\Pr[h(\x) = h(\y)] = 1 - \arccos(C(\x, \y))/\pi$, which can be used to derive a solution for cosine similarity, 
although not the clean solution that we could have hoped for in the style of MinHash for Jaccard similarity.
There exists a number of different data-independent LSH approaches~\cite{terasawa2007spherical, andoni2014beyond, andoni2015practical} that improve upon the $\rho$-value of SimHash.
Perhaps surprisingly, it turns out that these approaches yield lower $\rho$-values for the $(j_1, j_2)$-approximate Jaccard similarity search problem compared to MinHash for certain combinations of $(j_1, j_2)$.  
Unfortunately, while asymptotically superior these techniques suffer from a non-trivial $o_{n}(1)$-term in the exponent that only decreases very slowly with $n$.
In comparison, both MinHash and SimHash are simple to describe and have closed expressions for their $\rho$-values.
Furthermore, MinHash and SimHash both have the advantage of being efficient in the sense that a hash function can be represented using space $O(d)$ and the time to compute $h(\x)$ is  $O(|\x|)$. 

%
\begin{table*}[ht]
\caption{Overview of $\rho$-values for similarity search with Hamming vectors of equal weight $t$.}
\label{tab:comparison}
\renewcommand\arraystretch{2}
\begin{center}
\begin{tabular}{|l|c|c|c|}
\hline
{\bf\diagbox{Ref.}{Measure}} & 
\thead{Hamming\\ $r_1 < r_2$} &
\thead{Braun-Blanquet\\$b_1>b_2$} &  
\thead{Jaccard\\$j_1>j_2$} \\ 
\hline
\hline
Bit-sampling LSH~\cite{indyk1998} & 
$r_1/r_2$ &
$\tfrac{1-b_1}{1-b_2}$ &
$\tfrac{1-j_1}{1+j_1}/\tfrac{1-j_2}{1+j_2}$ \\ 
\hline
Minhash LSH~\cite{Bro97b} & 
$\log \tfrac{1 - r_{1}}{1 + r_{1}} / \log \tfrac{1 - r_{2}}{1 + r_{2}}$ &
$\log\tfrac{b_1}{2-b_1} / \log\tfrac{b_2}{2-b_2}$ &
$\log(j_1)/\log(j_2)$ \\
\hline
Angular LSH~\cite{andoni2015practical} & 
$\frac{r_1}{r_2}\frac{1 - r_{2}/2}{1 - r_{1}/2}$ &
$\tfrac{1-b_1}{1+b_1}/\tfrac{1-b_2}{1+b_2}$ &
$\tfrac{1-j_1}{1+3j_1}/\tfrac{1-j_2}{1+3j_2}$ \\
\hline
Data-dep.~LSH~\cite{andoni2015optimal} & 
$\frac{r_1}{r_2}\frac{1}{2 - r_{1}/r_{2}}$ &
$\tfrac{1-b_1}{1 + b_1 - 2 b_2}$ & 
$\frac{(1-j_{1})(1+j_{2})}{1 - j_{1}j_{2} + 3(j_{1}-j_{2})}$ \\
\hline
{\bf Theorem~\ref{thm:upper}} & 
$\log(1-r_{1}) / \log(1-r_{2})$ &
$\log(b_1)/\log(b_2)$ &
$\log\tfrac{2j_1}{1+j_1}/\log\tfrac{2j_2}{1+j_2}$ \\
\hline
\end{tabular}
\end{center}
\begin{minipage}[c]{0.72\textwidth}
\footnotetext{
{\bf Notes:} While most results in the literature are stated for a single measure, 
the fixed weight restriction gives a 1-1 correspondence that makes it possible to express the results in terms of other similarity measures.
Hamming distances are normalized by a factor $2t$ to lie in $[0;1]$. 
Lower order terms of $\rho$-values are suppressed, and for bit-sampling LSH we assume that the Hamming distances 
are small relative to the dimensionality of the space, i.e., that $2r_{1}t/d = o(1)$.} 
\end{minipage}
\end{table*}

In Table \ref{tab:comparison} we show how the upper bounds for similarity search under different measures of set similarity relate to each other in the case where all sets are $t$-sparse.
In addition to Hamming distance and Jaccard similarity, we consider Braun-Blanquet similarity~\cite{braunblanquet1932} defined as
\begin{equation}\label{eq:B}
B(\x, \y) = |\x \cap \y| / \max(|\x|, |\y|), 
\end{equation}
which for $t$-sparse vectors is identical to cosine similarity.
When the query and the sets in $P$ can have different sizes the picture becomes muddled, 
and the question of which of the known algorithms is best for each measure $S$ is complicated. 
In Section \ref{sec:equivalence} we treat the problem of different set sizes and provide a brief discussion for Jaccard similarity, 
specifically in relation to our upper bound for Braun-Blanquet similarity.

Similarity search under set similarity and the batched version often referred to as \emph{set similarity join}~\cite{arasu2006efficient,bayardo2007scaling}
have also been studied extensively in the information retrieval and database literature, but mostly without providing theoretical guarantees on performance.
Recently the notion of containment search, where the similarity measure is the (unnormalized) intersection size, was studied in the LSH framework~\cite{shrivastava2015asymmetric}.
This is a special case of \emph{maximum inner product} search~\cite{shrivastava2015asymmetric,ahle2015}.
However, these techniques do not give improvements in our setting.

\medskip

\para{Similarity Estimation.}
Finally, we mention that another application of MinHash~\cite{Bro97b,Bro97a} is the (easier) problem of \emph{similarity estimation}, 
where the task is to condense each vector $\x$ into a short signature $s(\x)$ in such a way that the similarity $J(\x,\y)$ can be estimated from $s(\x)$ and~$s(\y)$.
A related similarity estimation technique was independently discovered by Cohen~\cite{cohen1997size}.
Thorup~\cite{thorup2013bottom} has shown how to perform similarity estimation using just a small amount of randomness in the definition of the function $s(\cdot)$.
In another direction, Mitzenmacher et al.~\cite{mitzenmacher2014} showed that it is possible to improve the performance of MinHash for similarity estimation when the Jaccard similarity is close to~1,
but for smaller similarities it is known that succinct encodings of MinHash such as the one in~\cite{li2011theory} comes within a constant factor of the optimal space for storing~$s(\x)$~\cite{pagh2014min}.
Curiously, our improvement to MinHash in the context of similarity \emph{search} comes when the similarity is neither too large nor too small.
Our techniques do not seem to yield any improvement for the similarity \emph{estimation} problem.

\subsection{Contribution}
We show the following upper bound for approximate similarity search under Braun-Blanquet similarity: 
\begin{theorem}\label{thm:upper}
	For every choice of constants \mbox{$0 < b_{2} < b_{1} < 1$} we can solve the $(b_{1}, b_{2})$-approximate similarity search problem under Braun-Blanquet similarity 
	with query time $O(|\q|n^{\rho} \log n)$ and space usage $O(n^{1+\rho}\log n + \sum_{\x \in P}|\x|)$ where $\rho = \log(1/b_1) / \log(1/b_2)$.
\end{theorem}
In the case where the sets are $t$-sparse our Theorem \ref{thm:upper} gives the first strict improvement on the $\rho$-value 
for approximate Jaccard similarity search compared to the data-independent LSH approaches of MinHash and Angular LSH.
Figure~\ref{fig:two-approx} shows an example of the improvement for a slice of the parameter space.
The improvement is based on a new locality-sensitive mapping that considers a specific random collection of length-$k$ paths on the vertex set $\{1,\dots,d\}$, 
and associates each vector $\x$ with the paths in the collection that only visits vertices in $\{ i \; | \; \x_i=1 \}$.
Our data structure exploits that similar vectors will be associated with a common path with constant probability, while vectors with low similarity have a negligible probability of sharing a path.
However, the collection of paths has size superlinear in~$n$, so an efficient method is required for locating the paths associated with a particular vector.
Our choice of the collection of paths balances two opposing constraints:
It is random enough to match the filtering performance of a truly random collection of sets, and at the same time it is structured enough to allow efficient search for sets matching a given vector.
The search procedure is comparable in simplicity to the classical techniques of bit sampling, MinHash, SimHash, and $p$-stable LSH, and we believe it might be practical.
This is in contrast to most theoretical advances in similarity search in the past ten years that suffer from $o(1)$ terms in the exponent of complexity bounds.

\begin{figure}   
	\includegraphics[width=\columnwidth]{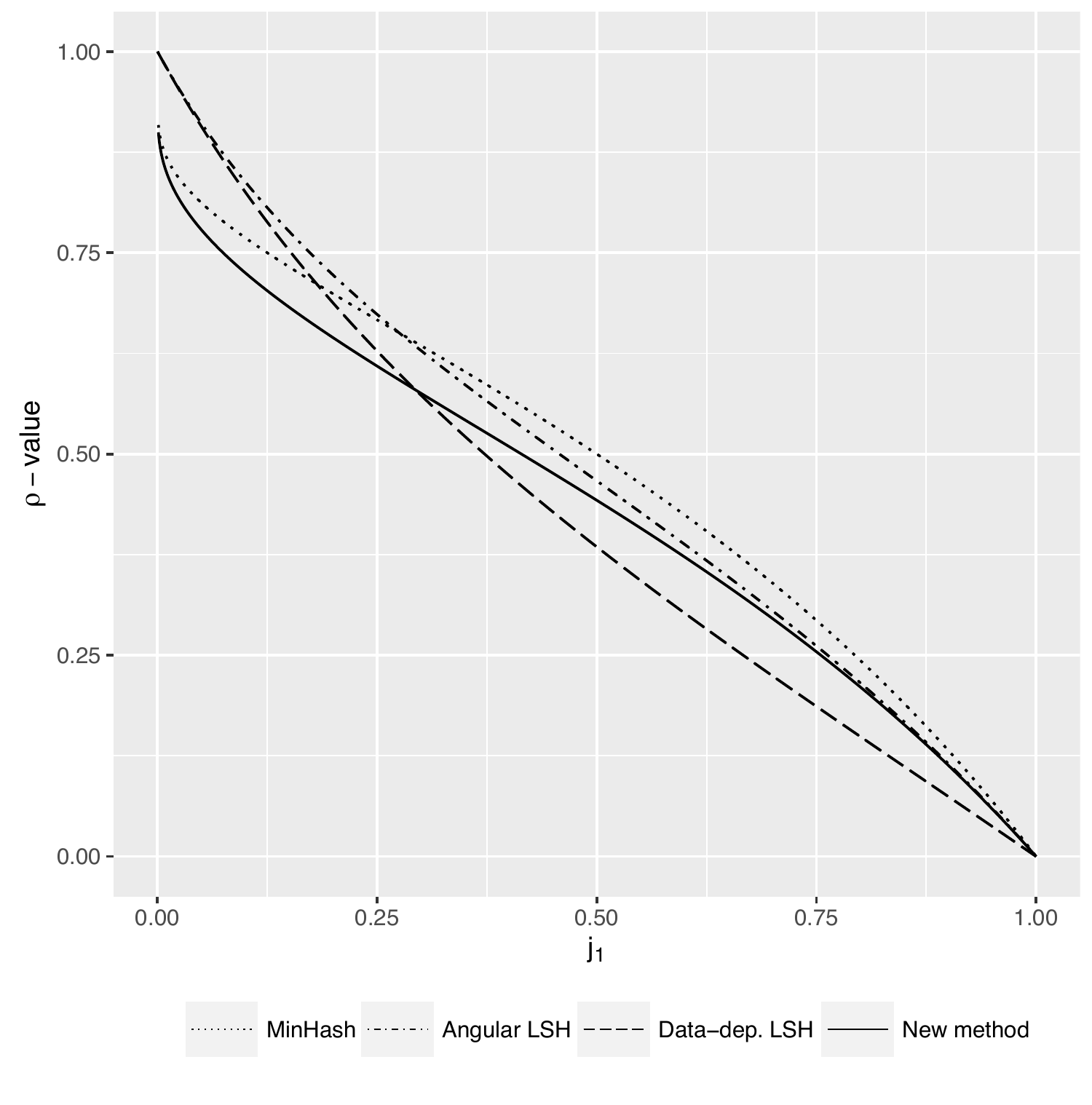}
	\caption{Exponent when searching for a vector with Jaccard similarity~$j_{1}$ with approximation factor~2 (i.e., guaranteed to return a vector with Jaccard similarity $j_{1}/2$) for various methods in the setting where all sets have the same size. 
	Our new method is the best data-independent method, and is better than data-dependent LSH up to about $j_1 \approx 0.3$.}
\label{fig:two-approx}
\end{figure}

\medskip

\para{Intuition.} 
Recall that we will think of a vector $\x \in \cube{d}$ also as a set, $\{ i \;|\; \x_i=1\}$.
MinHash can be thought of as a way of sampling an element $i_{\x}$ from $\x$, namely, we let $i_{\x} = \argmin_{i\in\x} h(i)$ where $h$ is a random hash function.
For sets $\x$ and $\y$ the probability that $i_{\x}=i_{\y}$ equals their Jaccard similarity $J(\x,\y)$, which is much higher than if the samples had been picked independently.
Consider the case in which $|\x|=|\y|=t$, so $J(\x,\y) = \frac{|\x\cap\y|}{2t-|\x\cap\y|}$.
Another way of sampling is to compute $I_{\x} = \x \cap \vect{b}$, where $\Pr[i\in \vect{b}]=1/t$, independently for each $i\in [d]$.
The expected size of $I_{\x}$ is 1, so this sample has the same expected ``cost'' as the MinHash-based sample.
But if the Jaccard similarity is small, the latter samples are more likely to overlap: 
$$\Pr[I_{\x}\cap I_{\y} \neq \emptyset] = 1-(1-1/t)^{|\x\cap\y|} \approx 1 - e^{-|\x \cap \y|/t} \approx |\x\cap\y|/t,$$ 
nearly a factor of 2 improvement.
In fact, whenever $|\x\cap\y| < 0.6\, t$ we have 
$\Pr[I_{\x}\cap I_{\y} \ne \emptyset] > \Pr[i_{\x}=i_{\y}]$.
So in a certain sense, MinHash is not the best way of collecting evidence for the similarity of two sets.
(This observation is not new, and has been made before e.g.~in~\cite{cohen2009leveraging}.)

\medskip

\para{Locality-Sensitive Maps.}
The intersection of the samples $I_{\x}$ does not correspond directly to hash collisions, so it is not clear how to turn this insight into an algorithm in the LSH framework.
Instead, we will consider a generalization of both the locality sensitive filtering (LSF) and LSH frameworks where we define a distribution $\LSM$ over maps $M \colon \cube{d} \to 2^{R}$.
The map $M$ performs the same task as the LSH data structure: It takes a vector $\x$ and returns a set of memory locations $M(\x) \subseteq \{1, \dots, R \}$.
A randomly sampled map $M \sim \LSM$ has the property that if a pair of points $\x,\y$ are close then $M(\x) \cap M(\y) \neq \emptyset$ with constant probability, 
while if $\x, \y$ are distant then the expected size $M(\x) \cap M(\y)$ is small (much smaller than $1$).
It is now straightforward to see that this distribution can be used to construct a data structure for similarity search by storing each data point $\x \in P$ in the set of memory locations or buckets $M(\x)$.
A query for a point $\y$ is performed by computing the similarity between $\y$ and every point $\x$ contained in the set buckets $M(\y)$, reporting the first sufficiently similar point found.

\medskip

\para{\textsc{Chosen Path}.}
It turns out that to most efficiently filter out vectors of low similarity in the setting where all sets have equal size, 
we would like to map each data point $\x \in \cube{d}$ to a collection $M(\x)$ of random subsets of $\cube{d}$ that are contained in $\x$.
Furthermore, to best distuinguish similar from dissimilar vectors when solving the approximate similarity search problem, we would like the random subsets of $\cube{d}$ to have size $\Theta(\log n)$.
This leads to another obstacle: The collection of subsets of $\cube{d}$ required to ensure that $M(\x) \cap M(\y) \neq \emptyset$ for similar points,
i.e., that $M$ maps to a subset contained in $\x \cap \y$, is very large.
The space usage and evaluation time of a locality-sensitive map $M$ to fully random subsets of $\cube{d}$ would far exceed $n$, rendering the solution useless. 
To overcome this we create the samples in a gradual, correlated way using a pairwise independent branching process that turns out to yield ``sufficiently random'' samples for the argument to go through.

\medskip

\para{Lower Bound.}
On the lower bound side we show that our solution for Braun-Blanquet similarity is best possible in terms of parameters $b_1$ and $b_2$ within the class of solutions that can be characterized as data-independent locality-sensitive maps.
The lower bound works by showing that a family of locality-sensitive maps for Braun-Blanquet similarity with a $\rho$-value below $\log(1/b_1)/\log(1/b_2)$ 
can be used to construct a locality-sensitive hash family for the $c$-approximate near neighbor problem in Hamming space with a $\rho$-value below $1/c$, 
thereby contradicting the LSH lower bound by O'Donnell et al.~\cite{odonnell2014optimal}.
We state the lower bound here in terms of locality-sensitive hashing, formally defined in Section \ref{sec:preliminaries}. 
\begin{theorem} \label{thm:lower}
For every choice of constants $0 < {b_2} < {b_1} < 1$ any $({b_1}, {b_2}, p_1, p_2)$-sensitive hash family $\LSH_{B}$ for $\cube{d}$ under Braun-Blanquet similarity must satisfy
\begin{equation*}
	\rho(\LSH_{B}) = \frac{\log(1/p_{1})}{\log(1/p_{2})} \geq \frac{\log(1/{b_1})}{\log(1/{b_2})} - O\left(\frac{\log(d/p_2)}{d}\right)^{1/3}. 
\end{equation*}
\end{theorem}
The details showing how this LSH lower bound implies a lower bound for locality-sensitive maps are given in Section \ref{sec:lower}. 
\section{Preliminaries}\label{sec:preliminaries}
As stated above we will view $\x \in \cube{d}$ both as a vector and as a subset of $[d] = \{1,\dots,d\}$.
Define $\x$ to be \emph{$t$-sparse} if $|\x| = t$; 
we will be interested in the setting where $t \leq d/2$, and typically the sparse setting $t \ll d$.
Although many of the concepts we use hold for general spaces, for simplicity we state definitions in the same setting as our results: 
the boolean hypercube $\cube{d}$ under some measure of similarity $S \colon \cube{d} \times \cube{d} \rightarrow [0;1]$. 

\begin{definition} (Approximate similarity search)
	Let $P \subset \cube{d}$ be a set of $|P| = n$ data vectors, let $S \colon \cube{d} \times \cube{d} \rightarrow [0;1]$ be a similarity measure, and let $s_1, s_2 \in [0;1]$ such that $s_1 > s_2$.
	A solution to the \emph{$(s_1,s_2)$-$S$-similarity search problem} is a data structure that supports the following query operation: 
	on input $\q \in \cube{d}$ for which there exists a vector $\x \in P$ with $S(\x,\q) \geq s_1$, return $\x' \in P$ with $S(\x',\q) > s_2$.
\end{definition}
Our data structures are randomized, and queries succeed with probability at least~$1/2$ (the probability can be made arbitrarily close to~$1$ by independent repetition).
Sometimes similarity search is formulated as searching for vectors that are near $\q$ according to the distance measure $D(\x, \y) = 1 - S(\x,\y)$. 
For our purposes it is natural to phrase conditions in terms of similarity, but we  compare to solutions originally described as ``near neighbor'' methods.

Many of the best known solutions to approximate similarity search problems are based on a technique of randomized space partitioning. 
This technique has been formalized in the locality-sensitive hashing framework \cite{indyk1998} and the closely related locality-sensitive filtering framework~\cite{becker2016, christiani2017framework}.
\begin{definition} (Locality-sensitive hashing {\cite{indyk1998}})
	A $({s_1}, {s_2}, p_1, p_2)$-sensitive family of hash functions for a similarity measure $S$ on $\cube{d}$ is a distribution $\LSH_{S}$ over functions $h \colon \cube{d} \to R$ 
	such that for all $\x,\y \in \cube{d}$ and random $h$ sampled according to $\LSH_{S}$: If $S(\x,\y) \geq s_1$ then $\Pr[h(\x) = h(\y)] \geq p_1$, and if $S(\x,\y) \leq s_2$ then $\Pr[h(\x) = h(\y)] \leq p_2$.
\end{definition}
The range $R$ of the family will typically be fairly small such that an element of $R$ can be represented in a constant number of machine words.
In the following we assume for simplicity that the family of hash functions is \emph{efficient} such that $h(\x)$ can be computed in time~$O(|\x|)$.
Furthermore, we will assume that the time to compute the similarity $S(\x, \y)$ can be upper bounded by the time it takes to compute the size of the intersection of preprocessed sets, i.e.,~$O(\min(|\x|, |\y|))$.

Given a locality-sensitive family it is quite simple to obtain a solution to the approximate similarity search problem, 
essentially by hashing points to buckets such that close points end up in the same bucket while distant points are kept apart. 
\begin{lemma}[LSH framework {\cite{indyk1998, har-peled2012}}] 
	Given a $(s_1, s_2, p_1, p_2)$-sensitive family of hash functions it is possible to solve the $(s_1, s_2)$-$S$-similarity search problem
	with query time $O(|\q| n^{\rho} \log n)$ and space usage $O(n^{1 + \rho} + \sum_{\x \in P}|\x|)$ where $\rho = \log(1/p_1) / \log(1/p_2)$.
\end{lemma}
The upper bound presented in this paper does not quite fit into the existing frameworks. 
However, we would like to apply existing LSH lower bound techniques to our algorithm.
Therefore we define a more general framework that captures solutions constructed using the LSH and LSF framework, as well as the upper bound presented in this paper.
\begin{definition}[Locality-sensitive map]
	A $(s_1, s_2, m_1, m_2)$-sensitive family of maps for a similarity measure $S$ on $\cube{d}$ is a distribution $\LSM_{S}$ 
	over mappings $M \colon \cube{d} \to 2^{R}$ (where $2^{R}$ denotes the power set of $R$) such that for all $\x, \y \in \cube{d}$ and random $M \in \LSM_{S}$:
	\begin{enumerate}
		\item $\E[|M(\x)|] \leq m_1$.	
		\item If $S(\x, \y) \leq s_2$ then $\E[|M(\x) \cap M(\y)|] \leq m_2$.
		\item If $S(\x, \y) \geq s_1$ then $\Pr[M(\x) \cap M(\y) \neq \emptyset] \geq 1/2$.
	\end{enumerate}
 \end{definition}
Once we have a family of locality-sensitive maps $\LSM$ we can use it to obtain a solution to the $(s_1, s_2)$-$S$-similarity search problem.
\begin{lemma}
	Given a $(s_1, s_2, m_1, m_2)$-sensitive family of maps $\LSM$ we can solve the $(s_1, s_2)$-$S$-similarity search problem 
	with query time $O(m_{1} + nm_{2}|\q| + T_{M})$ and space usage $O(nm_{1} + \sum_{\x \in P}|\x|)$ where $T_{M}$ is the time to evaluate a map $M \in \LSM$.
\end{lemma}
\begin{proof}
	We construct the data structure by sampling a map $M$ from $\LSM$ and use it to place points in $P$ into buckets.
	To run a query for a point $\q$ we proceed by evaluating $M(\q)$ and computing the similarity between $\q$ and the points in the buckets associated with $M(\q)$.
	If a sufficiently similar point is found we return it. 
	We get rid of the expectation in the guarantees by independent repetitions and applying Markov's inequality.
\end{proof}

\para{Model of Computation.}
We assume the standard word RAM model \cite{hagerup1998} with word size $\Theta(\log n)$, where $n=|P|$.
In order to be able to draw random functions from a family of functions we augment the model with an instruction that generates a machine word uniformly at random in constant time.
\section{Upper Bound} \label{sec:upper}
We will describe a family of locality-sensitive maps $\LSM_{B}$ for solving the $(b_1, b_2)$-$B$-similarity search problem, where $B$ is Braun-Blanquet similarity~(\ref{eq:B}).
After describing $\LSM_{B}$ we will give an efficient implementation of $M \in \LSM_{B}$ and show how to set parameters to obtain our Theorem \ref{thm:upper}. 
\subsection{Chosen Path}
The \textsc{Chosen Path} family $\LSM_{B}$ is defined by $k$ random hash functions $h_1, \dots, h_k$ where $h_{i} \colon [w] \times [d]^{i} \to [0;1]$ and $w$ is a positive integer.
The evaluation of a map $M_{k} \in \LSM_{B}$ proceeds in a sequence of $k+1$ steps that can be analyzed as a Galton-Watson branching process, 
originally devised to investigate population growth under the assumption of identical and independent offspring distributions. 
In the first step $i = 0$ we create a population of $w$ starting points
\begin{equation}\label{eq:M0}
M_{0}(\x) = [w].
\end{equation}
In subsequent steps, every path that has survived so far produces offspring according to a random process that depends on $h_i$ and the element $\x \in \cube{d}$ being evaluated. 
We use $p \circ j$ to denote concatenation of a path $p$ with a vertex $j$. 
\begin{equation}\label{eq:Mi}
M_i(\x) = \left\{ p \circ j \mid p \in M_{i-1}(\x) \land h_i(p\circ j) < \frac{\x_{j}}{b_{1}|\x|} \right\}. 
\end{equation}
Observe that $h_i(p\circ j) < \frac{\x_{j}}{b_{1}|\x|}$ can only hold when $\x_{j}=1$, so the paths in $M_i(\x)$ are constrained to $j\in\x$.
The set $M(\x) = M_{k}(\x)$ is given by the paths that survive to the $k$th step.
We will proceed by bounding the evaluation time of $M \in \LSM_{B}$ as well as showing the locality-sensitive properties of~$\LSM_{B}$.
In particular, for similar points $\x, \y \in \cube{d}$ with $B(\x, \y) \geq b_1$ we will show that with probability at least $1/2$ there will be a path that is chosen by both $\x$ and $\y$. 
\begin{lemma}[Properties of {\textsc{Chosen Path}}]\label{lem:cp}
	For all $\x, \y \in \cube{d}$, integer $i\geq 0$, and random $M \in \LSM_{B}$:
	\begin{enumerate}
		\item $\E[|M_{i}(\x)|] \leq (1/b_{1})^{i}w$.	
		\item If $B(\x, \y) < b_2$ then $\E[|M_{i}(\x) \cap M_{i}(\y)|] \leq (b_{2}/b_{1})^{i}w$.
		\item If $B(\x, \y) \geq b_1$ then $\Pr[M_{i}(\x) \cap M_{i}(\y) \neq \emptyset] \geq i/(i + w)$.
	\end{enumerate}
\end{lemma}
\begin{proof}
We prove each property by induction on $i$. 
The base cases $i=0$ follow from (\ref{eq:M0}). 
Now consider the inductive step for property 1.
Let $\1\{\mathcal{P}\}$ denote the indicator function for predicate~$\mathcal{P}$. Using independence of the hash functions $h_i$ we get:
\begin{align*}
	\E[|M_{i}(\x)|] 
	&= \E\left[ \sum_{p \in M_{i-1}(\x)} \sum_{j \in [d]} \1\! \left\{ h_{i}(p\circ j) < \frac{\x_{j}}{b_{1} |\x|} \right\} \right] \\ 
	&= \E\left[ \sum_{p \in M_{i-1}(\x)} 1 \right] \E\left[\sum_{j \in [d]} \1\! \left\{ h_{i}(p\circ j) < \frac{\x_{j}}{b_{1} |\x|} \right\} \right] \\ 
				    &\leq \E[|M_{i-1}(\x)|] /b_{1} \\
					&\leq (1/b_{1})^{i}w \enspace .
\end{align*}
The last inequality uses the induction hypothesis.
We use the same approach for the second property where we let $X_{i}=M_{i}(\x) \cap M_{i}(\y)$.
\begin{align*}
	\E[|X_{i}|] 
	&= \E\left[ \sum_{p \in X_{i-1}} \sum_{j \in [d]} \1\! \left\{ h_{i}(p\circ j) < \frac{\x_{j}}{b_{1} |\x|} \land h_{i}(p\circ j) < \frac{\y_{j}}{b_{1} |\y|}  \right\} \right] \\ 
	&= \E\left[ \sum_{p \in X_{i-1}} 1 \right] \sum_{j \in [d]} \Pr\left[ h_{i}(p\circ j) < \frac{\min(\x_{j},\y_{j})}{b_{1} \max(|\x|,|\y|)} \right] \\ 
			    &\leq \E[|X_{i-1}|] (B(\x, \y)/b_{1}) \\
				&\leq (B(\x, \y)/b_{1})^{i}w \enspace . 	
\end{align*}

To prove the third property we bound the variance of $|X_{i}|$ and apply Chebyshev's inequality to bound the probability of $X_{i}=\emptyset$.
First consider the case where $|\x| \leq 1/b_{1}$ and $|\y| \leq 1/b_1$. 
Here it must hold that $X_{i}>0$ as intersecting paths exist ($b_1>0$) and always activate. 
In all other cases we have that 
$$\E[|X_{i}|] = (B(\x, \y)/b_{1})^{i}w \enspace .$$
Knowing the expected value we can apply Chebyshev's inequality once we have an upper bound for $\Var[|X_{i}|] = \E[|X_{i}|^2] - \E[|X_{i}|]^{2}$.
Specifically we show that $\E[|X_{i}|^2]\leq wi (B(\x, \y)/b_{1})^{2i}$, by induction on $i$.
To simplify notation we define the indicator variable
$$Y_{p,j} = \1\! \left\{ h_i(p\circ j) < \frac{\x_{j}}{b_{1} |\x|} \land h_i(p\circ j) < \frac{\y_{j}}{b_{1} |\y|} \right\}$$
where we suppress the subscript~$i$.
First observe that
$$\E[Y_{p,j}]= 1 / (b_1 \max(|\x|,|\y|))\enspace .$$
By (\ref{eq:Mi}) we see that $|X_i| = \sum_{p\in X_{i-1}} \sum_{j\in [d]} Y_{p,j}$,
which means:
\begin{align*}
	\E[|X_{i}|^{2}] &= \E\left[ \left( \sum_{p \in X_{i-1}} \sum_{j \in [d]} Y_{p,j}  \right)^{2} \right] \\ 
   				    & = \E\left[ \sum_{p \in X_{i-1}} \sum_{j\in [d]} Y_{p,j}^{2}\right] \\
					&\qquad + \E\left[\sum_{p,p' \in X_{i-1}} \sum_{j,j'\in [d]} Y_{p,j} Y_{p',j'} \1\{(p,j)\ne (p',j')\} \right] \\
					&< \E[|X_{i-1}|] (B(\x, \y)/b_1) + \E[|X_{i-1}|^2] (B(\x, \y)/b_1)^2 \\
				    &\leq \sum_{s = 1}^{i}\E[|X_{i-s}|] (B(\x, \y)/b_1)^{2s-1} +   \E[|X_{0}|]^{2}(B(\x, \y)/b_1)^{2i} \\
					& = \E[|X_i|] \sum_{s = 0}^{i-1}(B(\x, \y)/b_{1})^{s} + \E[|X_i|]^2 \\
				    &\leq wi(B(\x,\y)/b_1)^{2i} +  \E[|X_i|]^2 \enspace .
\end{align*}
The third property now follows from a one-sided version of Chebychev's inequality applied to $|X_i|$.
\end{proof}
\subsection{Implementation Details}
Lemma \ref{lem:cp} continues to hold when the hash functions $h_1, \dots, h_k$ are individually \emph{2-independent} (and mutually independent) 
since we only use bounds on the first and second moment of the hash values.
We can therefore use a simple and practical scheme such as Zobrist hashing \cite{zobrist1970new} 
that hashes strings of $\Theta(\log n)$ bits to strings of $\Theta(\log n)$ bits in $O(1)$ time using space, say, $O(n^{1/2})$.
It is not hard to see that the domain and range of $h_1, \dots, h_k$ can be compressed to $O(\log n)$ bits (causing a neglible increase in the failure probability of the data structure). 
We simply hash the paths $p \in M_{i}(\x)$ to intermediate values of $O(\log n)$ bits, avoiding collisions with high probability, 
and in a similar vein, with high probability $O(\log n)$ bits of precision suffice to determine whether $h_i(p\circ j) < \frac{\x_{j}}{b_{1}|\x|}$. 

We now consider how to parameterize $\LSM_{B}$ to solve the $(b_{1}, b_{2})$-$B$-similarity problem on a set $P$ of $|P| = n$ points for every choice of constant parameters $0 < b_2 < b_1 < 1$, independent of $n$.
Note that we exclude $b_{1} = 1$ (which would correspond to identical vectors that can be found in time $O(1)$ by resorting to standard hashing) and $b_{2} = 0$ (for which every data point would be a valid answer to a query).
We set parameters 
\begin{align*}
k &= \lceil \log(n) / \log (1/b_{2}) \rceil, \\
w &= 2k
\end{align*}
from which it follows that $\LSM_{B}$ is $(b_1, b_2, m_1, m_2)$-sensitive with $m_1 = n^{\rho}w/b_1$ and $m_2 = n^{\rho - 1}w$ where $\rho = \log(1/b_1) / \log(1/b_2)$.
To bound the expected evaluation time of $M_{k}$ we use Zobrist hashing as well as intermediate hashes for the paths as described above.
In the $i$th step in the branching process the expected number of hash function evaluations is bounded by $|\q|$ times the number of paths alive at step $i - 1$. 
We can therefore bound the expected time to compute $M_{k}(\q)$ by
\begin{equation} \label{eq:hashtime}
	\sum_{i=0}^{k-1}\E[|\q||M_{i}(\q)|] \leq \frac{b_{1}^{-k}-1}{b_{1}^{-1}-1}|\q|w = O(|\q|n^{\rho}w).
\end{equation}
This completes the proof of Theorem \ref{thm:upper}.\footnote{We know of a way of replacing the multiplicative factor $|\q|$ in equation \eqref{eq:hashtime} by an additive term of $O(|\q|k)$ 
by choosing the hash functions $h_i$ carefully, but do not discuss this improvement here since $|\q|$ can be assumed to be polylogarithmic and our focus is on the exponent of $n$.}
\subsection{Comparison}\label{sec:comparison}
We will proceed by comparing our Theorem \ref{thm:upper} to results that can be achieved using existing techniques.
Again we focus on the setting where data points and query points are exactly $t$-sparse.
An overview of different techniques for three measures of similarity is shown in Table~\ref{tab:comparison}.
To summarize: The \textsc{Chosen Path} algorithm of Theorem \ref{thm:upper} improves upon all existing data-independent results over the entire $0 < b_2 < b_1 < 1$ parameter space.
Furthermore, we improve upon the best known \emph{data-dependent} techniques \cite{andoni2015optimal} for a large part of the parameter space (see Figure~\ref{fig:data}). 
The details of the comparisons are given in Appendix \ref{app:comparison}.

\medskip

\para{MinHash.} 
For $t$-sparse vectors there is a 1-1 mapping between Braun-Blanquet and Jaccard similarity. 
In this setting $J(\x,\y)=B(\x,\y)/(2-B(\x,\y))$.
Let $b_1 = 2j_1/(j_1+1)$ and $b_2 = 2j_2/(j_2+1)$ be the Braun-Blanquet similarities corresponding to Jaccard similarities $j_1$ and $j_2$.
The LSH framework using MinHash achieves $\rho_\text{minhash} = \log\left(\tfrac{b_1}{2-b_1}\right) / \log\left(\tfrac{b_2}{2-b_2}\right)$; 
this should be compared to $\rho = \log(b_1)/\log(b_2)$ achieved in Theorem~\ref{thm:upper}.
Since the function $f(z) = \log(\tfrac{z}{2-z})/\log z$ is monotonically increasing in $[0;1]$ we have that $\rho / \rho_\text{minhash} = f(b_2)/f(b_1) < 1$, i.e., $\rho$ is always smaller than $\rho_\text{minhash}$.
As an example, for $j_1 = 0.2$ and $j_2 = 0.1$ we get $\rho = 0.644...$ while $\rho_\text{minhash} = 0.698...$.
Figure~\ref{fig:minhash} shows the difference for the whole parameter space.
\begin{figure}
	\includegraphics[width=\columnwidth]{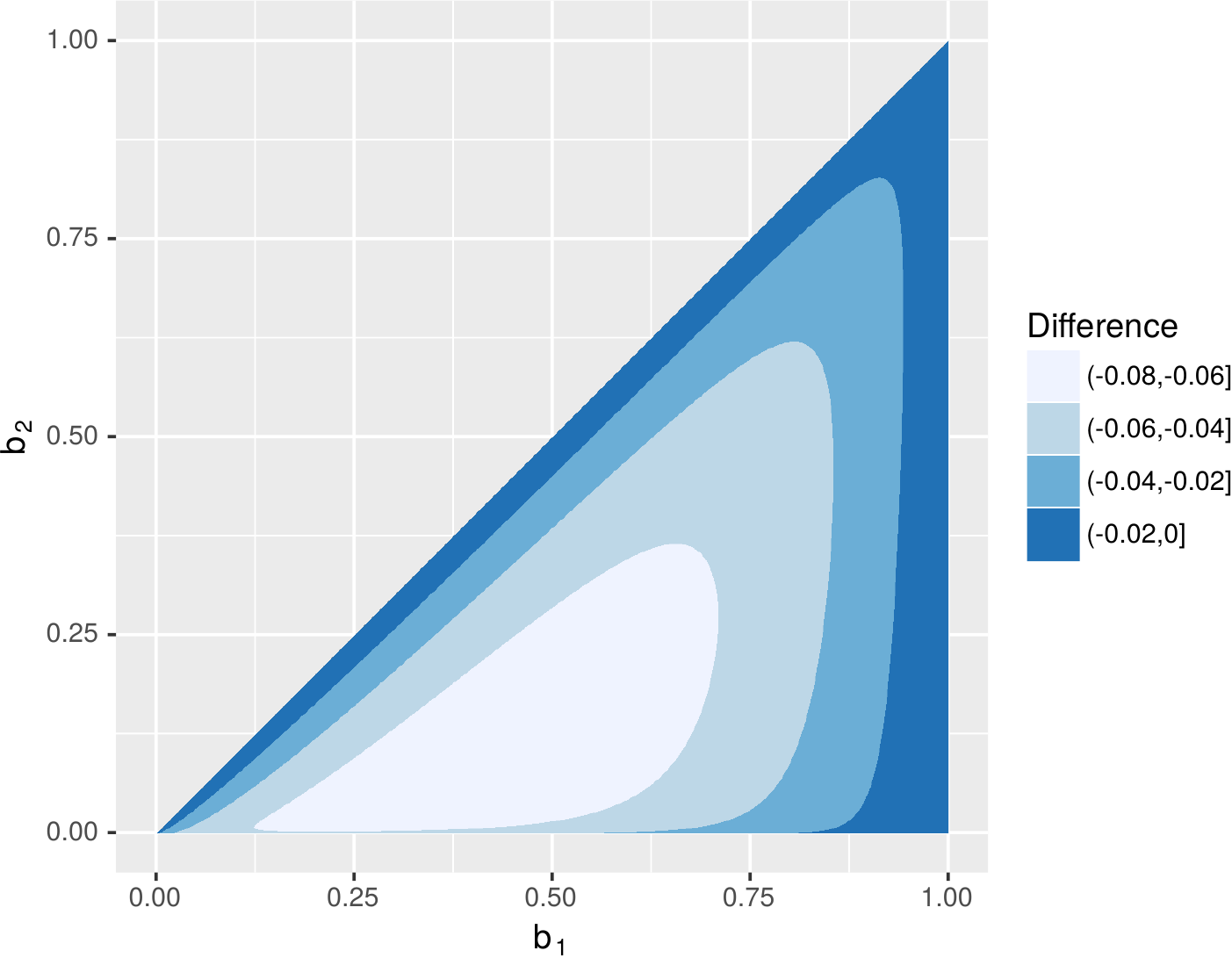}
	\caption{The difference $\rho - \rho_{\text{minhash}}$ comparing \textsc{Chosen Path} and MinHash in terms of Braun-Blanquet similarities $0 < b_2 < b_1 < 1$.}
	\label{fig:minhash}
\end{figure}
\medskip

\para{Angular LSH.}
Since our vectors are exactly $t$-sparse Braun-Blanquet similarities correspond directly to dot products (which in turn correspond to angles).
Thus we can apply angular LSH such as SimHash~\cite{charikar2002} or cross-polytope LSH~\cite{andoni2015practical}.
As observed in~\cite{christiani2017framework} one can express the $\rho$-value of cross-polytope LSH in terms of dot products as $\rho_\text{angular} = \tfrac{1-b_1}{1+b_1}/\tfrac{1-b_2}{1+b_2}$.
Since the function $f'(z) = (1 + z) \log(z)/(1 - z)$ is negative and monotonically increasing in $[0;1]$ we have that $\rho / \rho_\text{angular} = f'(b_1)/f'(b_2) < 1$, 
i.e., $\rho$ is always smaller than $\rho_\text{angular}$.
In the above example, for $j_1=0.2$ and $j_2=0.1$ we have $\rho_\text{angular} = 0.722...$ which is about $0.078$ more than {\sc Chosen Path}.
See Figure~\ref{fig:angular} for a visualization of the difference for the whole parameter space.
\begin{figure}
	\centering
	\includegraphics[width=\columnwidth]{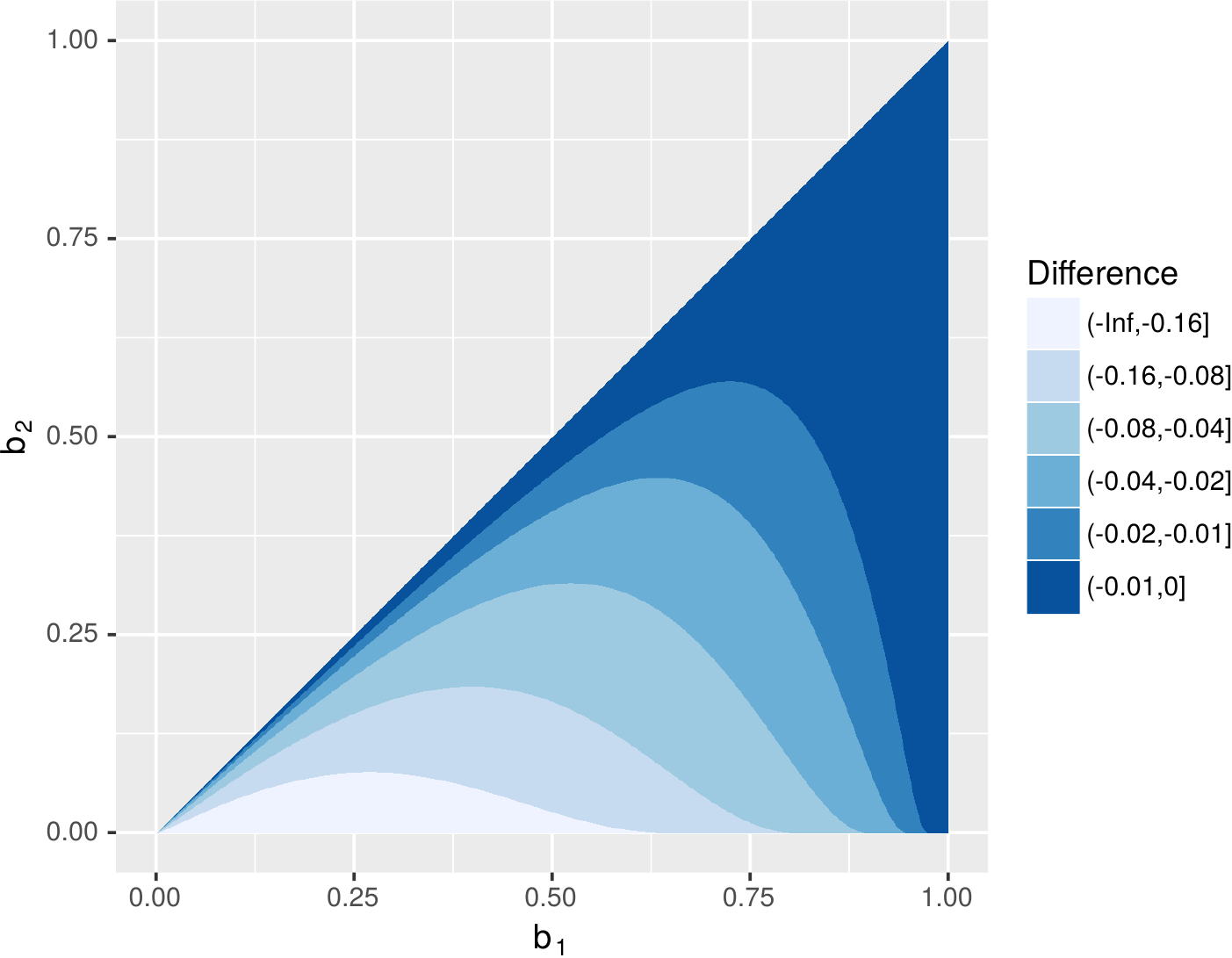}
	\caption{The difference $\rho - \rho_{\text{angular}}$ comparing \textsc{Chosen Path} and angular LSH in terms of Braun-Blanquet similarities $0 < b_2 < b_1 < 1$.}
	\label{fig:angular}
\end{figure}
\medskip

\para{Data-Dependent Hamming LSH.}
The Hamming distance between two $t$-sparse vectors with Braun-Blanquet similarity $b$ is $2t(1-b)$, since the intersection of the vectors has size $tb$.
This means that $(b_1,b_2)$-$B$-similarity search can be reduced to Hamming similarity search with approximation factor $c = (2t(1-b_1))/(2t(1-b_2)) = (1-b_1)/(1-b_2)$.
As mentioned above, the \emph{data dependent} LSH technique of~\cite{andoni2015optimal} achieves $\rho = 1/(2c-1)$ ignoring $o_{n}(1)$ terms. 
In terms of $b_1$ and $b_2$ this is $\rho_\text{datadep} = \frac{1-b_1}{1+b_1-2b_2}$, 
which in incomparable to the $\rho$ of Theorem \ref{thm:upper}.
In Appendix \ref{app:comparison} we show that $\rho < \rho_{\text{datadep}}$ whenever $b_2 \leq 1/5$, or equivalently, whenever $j_2 \leq 1/9$.
Revisiting the above example, for $j_1 = 0.2$ and $j_2 = 0.1$ we have $\rho_\text{datadep} = 0.6875$ which is about $0.043$ more than {\sc Chosen Path}.
Figure~\ref{fig:data} gives a comparison covering the whole parameter space.
\begin{figure}
	\centering
	\includegraphics[width=\columnwidth]{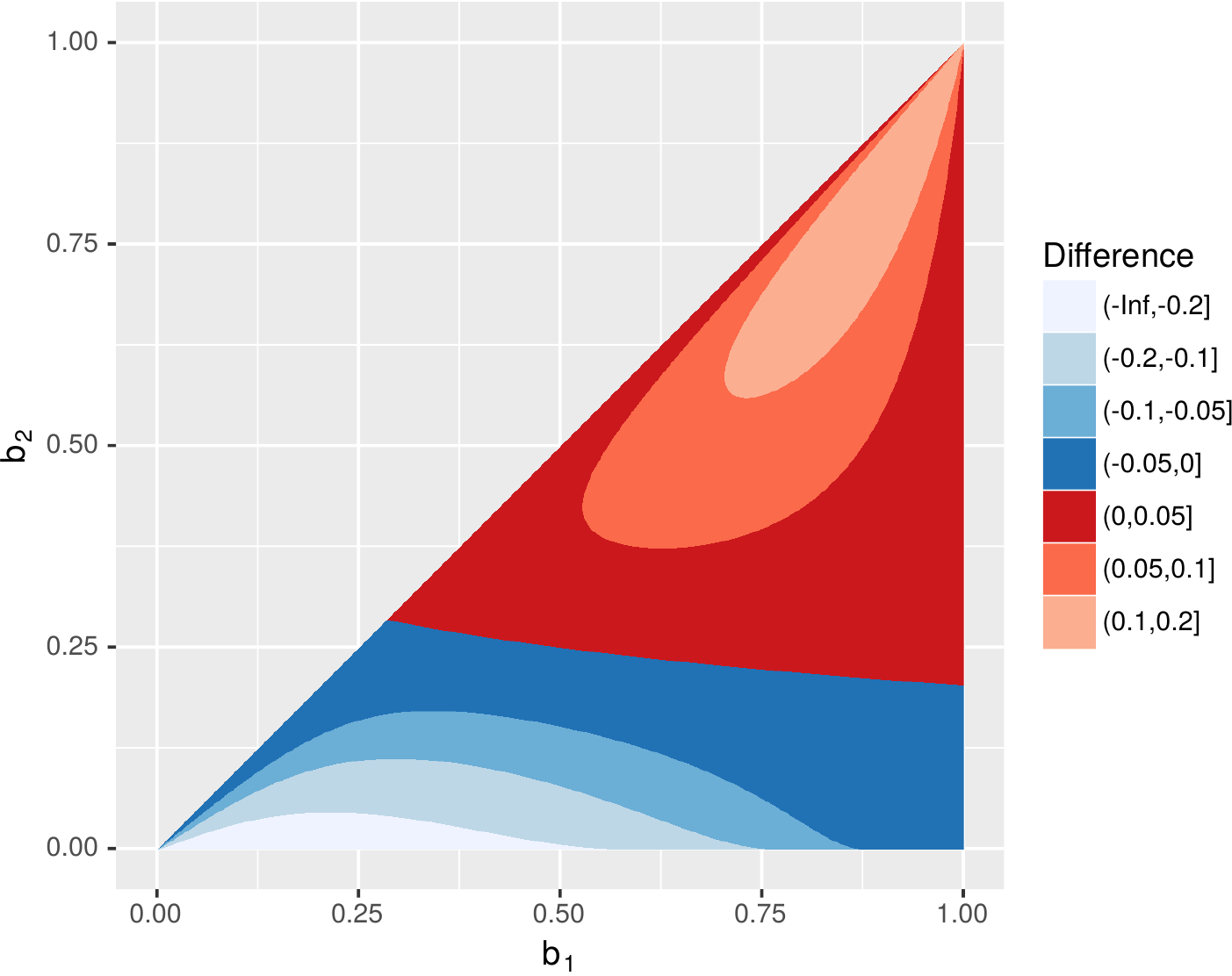}
	\caption{The difference $\rho - \rho_{\text{datadep}}$ comparing \textsc{Chosen Path} and data-dependent LSH in terms of Braun-Blanquet similarities $0 < b_2 < b_1 < 1$. 
	In the area of the parameter space that is colored blue we have that $\rho \leq \rho_{\text{datadep}}$ while for the red area it holds that $\rho > \rho_{\text{datadep}}$.}
	\label{fig:data}
\end{figure}

\section{Lower Bound} \label{sec:lower}
In this section we will show a locality-sensitive hashing lower bound for $\cube{d}$ under Braun-Blanquet similarity. 
We will first show that LSH lower bounds apply to the class of solutions to the approximate similarity search problem that are based on locality-sensitive maps, thereby including our own upper bound.
Next we will introduce some relevant tools from the literature, in particular the LSH lower bounds for Hamming space by O'Donnell et al. \cite{odonnell2014optimal} which we use, 
through a reduction, to show LSH lower bounds under Braun-Blanquet similarity.

\para{Lower Bounds For Locality-Sensitive Maps.}
Because our upper bound is based on a locality-sensitive map $\LSM_{B}$ and not LSH-based we first show that LSH lower bounds apply to LSM-based solutions. 
This is not too surprising as both the LSH and LSF frameworks produce LSM-based solutions.
We note that the idea of showing lower bounds for a more general class of algorithms that encompasses both LSH and LSF was used by Andoni et al.~\cite{andoni2017optimal} 
in their list-of-points data structure lower bound for the space-time tradeoff of solutions to the approximate near neighbor problem in the random data regime.
We use the approach of Christiani \cite{christiani2017framework} to convert an LSM family into an LSH family using MinHash. 

\begin{lemma}\label{lem:lsmtolsh}
	Suppose we have a $(s_1, s_2, m_1, m_2)$-sensitive family of maps $\LSM$ for a similarity measure $S$ on $\cube{d}$. 
	Then we can construct a $(s_1, s_2, p_1, p_2)$-sensitive family of hash functions~$\LSH$ for $S$ such that $p_1 = 1/8m$ and $p_2 = m_{2}/m$ where $m = \lceil 8 m_{1} \rceil$.
\end{lemma}
\begin{proof}
We sample a function $h$ from $\LSH$ by sampling a function $M$ from $\LSM$, modify $M$ to output a set of fixed size, and apply MinHash to the resulting set. 
For $M \in \LSM$ we define the function $\tilde{M}$ where we ensure that the size of the output set is $m$. 
We note that the purpose of this step is to be able to simultaneously lower bound $p_1$ and upper bound $p_2$ for $\LSH$ when we apply MinHash to the resulting sets.
\begin{equation*}
\tilde{M}(\x) =
\begin{cases}
\{ (\x, 1), \dots, (\x, m) \} & \text{if } |M(\x)| \geq m, \\
\{ (\x, 1), \dots, (\x, m - |M(\x)|) \} \cup M(\x) & \text{otherwise.} 
\end{cases}
\end{equation*}
We proceed by applying MinHash to the set $\tilde{M}(\x)$. Let $\pi$ denote a random permutation of the range of $\tilde{M}$ and define  
\begin{equation*}
	h(\x) = \argmin_{z \in \tilde{M}(\x)} \pi(z).
\end{equation*}
We then have 
\begin{equation*}
\Pr[h(\x) = h(\y)] = \sum_{\xi} \Pr[ J(\tilde{M}(\x), \tilde{M}(\y)) = \xi] \cdot \xi
\end{equation*}
summing over the finite set of all possible Jaccard similarities $\xi = a/b$ with $a, b \in \{0, 1, \dots, 2m \}$.
It is now fairly simple to lower bound $p_1$ and upper bound $p_2$. 
Assume that $\x, \y$ satisfy that $S(\x, \y) \geq s_1$. 
To lower bound $p_1$ we use a union bound together with Markov's inequality to bound the following probability:
\begin{align*}
	&\Pr[\tilde{M}(\x) \cap \tilde{M}(\y) = \emptyset] \\ 
	&\qquad \leq \Pr[M(\x) \cap M(\y) = \emptyset \wedge |M(\x)| \geq m \wedge |M(\y)| \geq m] \\
&\qquad \leq \Pr[M(\x) \cap M(\y) = \emptyset] + \Pr[|M(\x)| \geq m] + \Pr[|M(\y)| \geq m] \\
&\qquad \leq 1/2 + 1/8 + 1/8 
\end{align*}
We therefore have that $\Pr[\tilde{M}(\x) \cap \tilde{M}(\y) \neq \emptyset] \geq 1/4$. 
In the event of a nonempty intersection the probability of collision is given by $J(\tilde{M}(\x) \cap \tilde{M}(\y)) \geq 1/2m$ allowing us to conclude that $p_1 \geq 1/8m$.

Bounding the collision probability for distant pairs of points $\x, \y$ with $S(\x, \y) \leq s_2$ we get
\begin{equation*}
\sum_{\xi} \Pr[J(\tilde{M}(\x), \tilde{M}(\y)) = \xi] \cdot \xi \leq (1/m) \sum_{i = 1}^{\infty}\Pr[|\tilde{M}(\x) \cap \tilde{M}(\y)|] \cdot i = \frac{m_{2}}{m}.
\end{equation*}
\end{proof}
We are now ready to justify the statement that LSH lower bounds apply to LSM, allowing us to restrict our attention to proving LSH lower bounds for Braun-Blanquet similarity.
\begin{corollary}\label{cor:lsmtolsh}
Suppose that we have an LSM-based solution to the $(s_1, s_2)$-$S$-similarity search problem with query time $O(n^{\rho})$.
Then there exists a family $\LSH$ of locality-sensitive hash functions with $\rho(\LSH) = \rho + O(1/\log n)$.  
\end{corollary}
\begin{proof}
	The existence of the LSM-based solution implies that for every $n$ there exists a $(s_1, s_2, m_1, m_2)$-sensitive family of maps $\LSM$ with $m_1 = O(n^{\rho})$ and $nm_2 = O(n^{\rho})$.
	The upper bound on $\rho$ follows from applying Lemma \ref{lem:lsmtolsh}.  
\end{proof}

\para{LSH Lower Bounds for Hamming Space.}
There exist a number of powerful results that lower bound the $\rho$-value that is attainable by locality-sensitive hashing and related approaches in various settings 
\cite{motwani2007, panigrahy2010lower, odonnell2014optimal, andoni2016tight, christiani2017framework, andoni2017optimal}.
O'Donnell et al. \cite{odonnell2014optimal} showed an LSH lower bound of $\rho = \log(1/p_1) / \log(1/p_2) \geq 1/c - o_{d}(1)$
for $d$-dimensional Hamming space under the assumption that $p_2$ is not too small compared to $d$, i.e.,~$\log(1/p_2) = o(d)$.
The lower bound by O'Donnell et al.\ holds for $(r, cr, p_1, p_2)$-sensitive families for a particular choice of $r$ that depends on $d$, $p_2$, and $c$, 
and where $r$ is small compared to $d$ (for instance, we have that $r = \tilde{\Theta}(d^{2/3})$ when $c$ and $p_2$ are constant).

\medskip

We state a simplified version of the lower bound due to O'Donnell et al.~where $r = \sqrt{d}$ that we will use as a tool to prove our lower bound for Braun-Blanquet similarity. 
The full proof of Lemma \ref{lem:owzsimple} is given in Appendix \ref{app:lower}.
\begin{lemma} \label{lem:owzsimple}
	For every $d \in \mathbb{N}$, $1/d \leq p_2 \leq 1 - 1/d$, and $1 \leq c \leq d^{1/8}$ 
	every $(\sqrt{d}, c\sqrt{d}, p_1, p_2)$-sensitive hash family $\LSH$ for $\cube{d}$ under Hamming distance must have 
	\begin{equation}
		\rho(\LSH) = \frac{\log(1/p_{1})}{\log(1/p_{2})} \geq \frac{1}{c} - O(d^{-1/4}).
	\end{equation}
\end{lemma}
In general, good lower bounds for the entire parameter space $(r, cr)$ are not known, 
although the techniques by O'Donnell et al. appear to yield a bound of $\rho \gtrsim \log(1-2r/d)/\log(1-2cr/d)$.
This is far from tight as can be seen by comparing it to the bit-sampling \cite{indyk1998} upper bound of $\rho = \log(1-r/d)/\log(1-cr/d)$. 
Existing lower bounds are tight in two different settings. 
First, in the setting where $cr \approx d/2$ (random data), lower bounds \cite{motwani2007, dubiner2010bucketing, andoni2016tight} 
match various instantiations of angular LSH \cite{terasawa2007spherical, andoni2014beyond, andoni2015practical}. 
Second, in the setting where $r \ll d$, the lower bound by O'Donnell et al.~\cite{odonnell2014optimal} becomes $\rho \gtrsim \log(1-2r/d)/\log(1-2cr/d) \approx 1/c$, 
matching bit-sampling LSH~\cite{indyk1998} as well as Angular LSH.
\subsection{Braun-Blanquet LSH Lower Bound}
We are now ready to prove the LSH lower bound from Theorem \ref{thm:lower}.
The lower bound together with Corollary \ref{cor:lsmtolsh} shows that the $\rho$-value of Theorem \ref{thm:upper} 
is best possible up to $o_{d}(1)$ terms within the class of data-independent locality-sensitive maps for Braun-Blanquet similarity.
Furthermore, the lower bound also applies to angular distance on the unit sphere where it comes close to matching the best known upper bounds 
for much of the parameter space as can be seen from Figure~\ref{fig:angular}. 

\medskip

\para{Proof Sketch.}
The proof works by assuming the existence of a $(b_1, b_2, p_1, p_2)$-sensitive family $\LSH_{B}$ 
for $\cube{d}$ under Braun-Blanquet similarity with $\rho = \log(1/b_1)/\log(1/b_2) - \gamma$ for some $\gamma > 0$.
We use a transformation $T$ from Hamming space to Braun-Blanquet similarity to show that the existence of $\LSH_{B}$ implies the existence of a $(r, cr, p_{1}', p_{2}')$-sensitive 
family $\LSH_{H}$ for $D$-dimensional Hamming space that will contradict the lower bound of O'Donnell et al.~\cite{odonnell2014optimal} 
as stated in Lemma \ref{lem:owzsimple} for some appropriate choice of $\gamma = \gamma(d, p_{2})$.

We proceed by giving an informal description of a simple ``tensoring'' technique for converting a similarity search problem in Hamming space 
into a Braun-Blanquet set similarity problem for target similarity thresholds $b_1, b_2$.  
For $\x \in \cube{d}$ define 
$$\tilde{\x} = \{(i, \x_{i}) \mid i \in [d] \}$$ 
and for a positive integer $\tau$ define $\x^{\otimes \tau} = \{ (v_1, \dots, v_\tau) \mid v_i \in \tilde{\x} \}$.
We have that $|\x^{\otimes \tau}| = |\tilde{\x}|^{\tau} = d^\tau$ 
and 
$$B(\x^{\otimes \tau}, \y^{\otimes \tau}) = |\tilde{\x} \cap \tilde{\y}|^{\tau} / |\tilde{\x}|^\tau = (1 - r/d)^\tau$$ 
where $r = \norm{\x - \y}_{1}$.
For every choice of constants $0 < b_2 < b_1 < 1$ we can choose $d$, $\tau$, $r$, and $c \geq 1$ such that $(1 - r/d)^\tau \approx b_1$ and $(1 - cr/d)^\tau \approx b_2$.
Now, if there existed an LSH family for Braun-Blanquet with $\rho < \log(1/b_1)/\log(1/b_2)$ we would be able to obtain an LSH family for Hamming space with 
$$\rho < \log(1/b_1)/\log(1/b_2) = \log(1/(1 - r/d))/\log(1/(1-cr/d)) \leq 1/c .$$
For appropriate choices of parameters this would contradict the O'Donnell et al.~LSH lower bound of $\rho \gtrapprox 1/c$ for Hamming space. 
The proof itself is mostly an exercise in setting parameters and applying the right bounds and approximations to make everything fit together with the intuition above.
Importantly, we use sampling in order to map to a dimension that is much lower than the $d^\tau$ from the proof sketch in order to make the proof hold for small values of $p_2$ in relation to $d$. 

\medskip

\para{Hamming to Braun-Blanquet Similarity.}
Let $d \in \mathbb{N}$ and let $0 < b_2 < b_1 < 1$ be constant as in Theorem \ref{thm:lower}.
Let $\varepsilon \geq 1/d$ be a parameter to be determined.
We want to show how to use a transformation $T \colon \cube{D} \to \cube{d}$ from Hamming distance to Braun-Blanquet similarity
together with our family $\LSH_{B}$ to construct a $(r, cr, p_{1}', p_{2}')$-sensitive family $\LSH_{H}$ for $D$-dimensional Hamming space with parameters
\begin{align*}
D &= 2^d \\
r &= \sqrt{D} \\
c &= \frac{\ln(1/(b_2 - \varepsilon))}{\ln(1/(b_1 + \varepsilon))}
\end{align*}
where $p_{1}'$ and $p_{2}'$ remain to be determined.

The function $T$ takes as parameters positive integers $t$, $l$, and~$\tau$.
The output of $T$ consists of $t$ concatenated $l$-bit strings, each of of Hamming weight one.
Each of the $t$ strings is constructed independently at random according to the following process:
Sample a vector of indices $\vect{i} = (i_1, i_2, \dots, i_{\tau})$ uniformly at random from $[D]^{\tau}$ and define 
$\x_{\vect{i}} \in \cube{\tau}$ as $\x_{\vect{i}} = \x_{i_{1}} \circ \x_{i_{2}} \circ \dots \circ \x_{i_{\tau}}$.
Let $\z(\x) \in \cube{2^{\tau}}$ be indexed by $j \in \cube{\tau}$ and set the bits of $\z(\x)$ as follows: 
\begin{equation*}
	\z(\x)_{j} =
	\begin{cases}
		1 & \text{if } \x_{\vect{i}} = j, \\
		0 & \text{otherwise.} 
	\end{cases}
\end{equation*}
Next we apply a random function $g \colon \cube{\tau} \to [l]$ in order to map $\z(\x)$ down to an $l$-bit string $\vect{r}(\z(\x))$ 
of Hamming weight one while approximately preserving Braun-Blanquet similarity. For $i \in [l]$ we set
\begin{equation*}
	\vect{r}(\z(\x))_{i} = \bigvee_{j : g(j) = i} \z(\x)_{j}.
\end{equation*}
Finally we set 
\begin{equation*}
	T(\x) = \vect{r}_{1}(\z_{1}(\x)) \circ \vect{r}_{2}(\z_{2}(\x)) \circ \dots \circ \vect{r}_{t}(\z_{t}(\x))
\end{equation*}
where each $\vect{r}_{i}(\z_{i}(\x))$ is constructed independently at random.

We state the properties of $T$ for the following parameter setting:
\begin{align*}
\tau &= \lfloor \sqrt{D} \ln(1/(b_1 + \varepsilon)) \rfloor \\
l &= \lceil 8/\varepsilon \rceil \\
t &= \lfloor d/l \rfloor.
\end{align*}
\begin{lemma}
	For every $d \in \mathbb{N}$ and $D = 2^{d}$ there exists a distribution over functions of the form $T \colon \cube{D} \to \cube{d}$ such that for all $\x, \y \in \cube{D}$ and random $T$:
\begin{enumerate}
	\item $|T(\x)| = t$.
	\item If $\norm{\x-\y}_1 \leq r$ then $B(T(\x), T(\y)) \geq {b_1}$ with probability at least $1 - e^{-t\varepsilon^{2}/2}$.
	\item If $\norm{\x-\y}_1 > cr$ then $B(T(\x), T(\y)) < {b_2}$ with probability at least $1 - 2e^{t\varepsilon^{2}/32}$.
\end{enumerate}
\end{lemma}
\begin{proof}
The first property is trivial. 
For the second property we consider $\x, \y$ with $\norm{\x - \y}_1 \leq r$ where we would like to lower bound  
\begin{equation*}
B(T(\x), T(\y)) = \frac{|T(\x) \cap T(\y)|}{\max(|T(\x)|, |T(\y)|)}.
\end{equation*}
We know that $|T(\x)| = |T(\y)| = t$ so it remains to lower bound the size of the intersection $|T(\x) \cap T(\y)|$.
Consider the expectation
\begin{equation*}
	\E[|T(\x) \cap T(\y)|] = t\Pr[\z(\x) = \z(\y)].
\end{equation*}
We have that $\z(\x) = \z(\y)$ if $\x$ and $\y$ take on the same value in the $\tau$ underlying bit-positions that are sampled to construct $\z$.
Under the assumption that $\varepsilon \geq 1/d$, 
then for $d$ greater than some sufficiently large constant we can use a standard approximation to the exponential function (detailed in Lemma \ref{lem:exp} in Appendix \ref{app:lower}) to show that
\begin{align*}
	\Pr[\z(\x) = \z(\y)] &\geq (1 - r/D)^\tau \\ 
						 &\geq (1 - 1/\sqrt{D})^{\sqrt{D} \ln(1/(b_1 + \varepsilon))} \\
						 &\geq e^{\ln(b_1 + \varepsilon)}(1 - (\ln(b_1 + \varepsilon))^2 / \sqrt{D}) \\
						 &\geq b_1 + \varepsilon/2.
\end{align*}
Seeing as $|T(\x) \cap T(\y)|$ is the sum of $t$ independent Bernoulli trials we can apply Hoeffding's inequality to yield the following bound:
\begin{equation*}
	\Pr[|T(\x) \cap T(\y)| \leq b_1 t] \leq e^{-t \varepsilon^{2}/2}. 
\end{equation*}
This proves the second property of $T$.

For the third property we consider the Braun-Blanquet similarity of distant pairs of points $\x, \y$ with $\norm{\x - \y}_{1} > cr$.
Again, under our assumption that $\varepsilon \geq 1/d$ and for $d$ greater than some constant we have
\begin{align*}
	\Pr[\z(\x) = \z(\y)] &\leq (1 - cr/D)^\tau \\ 
																	  &\leq \frac{\left(1 - \frac{\ln(1/(b_2 - \varepsilon))}{\sqrt{D} \ln(1/(b_1 + \varepsilon))}\right)^{\sqrt{D} \ln(1/(b_1 + \varepsilon))}} {1 - c/\sqrt{D} } \\
						  &\leq (1 + 2c/\sqrt{D})(b_2 - \varepsilon) \\
						  &\leq b_2 - \varepsilon/2.
\end{align*}
There are two things that can cause the event $B(T(\x), T(\y)) < {b_2}$ to fail.
First, the sum of the $t$ independent Bernoulli trials for the event $\z(\x) = \z(\x')$ can deviate too much from its expected value.
Second, the mapping down to $l$-bit strings that takes place from $\z(\x)$ to $\vect{r}(\z(\x))$ can lead to an additional increase in the similarity due to collisions.
Let $Z$ denote the sum of the~$t$ Bernoulli trials for the events $\z(\x) = \z(\x')$ associated with $T$. 
We again apply a standard Hoeffding bound to show that
\begin{equation*}
	\Pr[Z \geq (b_2 - \varepsilon/4)t] \leq e^{-t\varepsilon^{2}/8}.
\end{equation*}
Let $X$ denote the number of collisions when performing the universe reduction to $l$-bit strings.
By our choice of $l$ we have that $E[X] \leq (\varepsilon/8)t$. Another application of Hoeffding's inequality shows that
\begin{equation*}
	\Pr[X \geq (\varepsilon/4)t] \leq e^{-t\varepsilon^{2}/32}.
\end{equation*}
We therefore get that
\begin{equation*}
	\Pr[|T(\x) \cap T(\x')| \geq b_2 t] \leq 2e^{-t \varepsilon^{2}/32}. 
\end{equation*}
This proves the third property of $T$.
\end{proof}
\para{Contradiction.}
To summarize, using the random map $T$ together with the LSH family $\LSH_{B}$ we can obtain an $(r, cr, {p_{1}}', {p_{2}}')$-sensitive family $\LSH_{H}$ for $D$-dimensional Hamming space 
with ${p_{1}}' = {p_{1}} - \delta$ and ${p_{2}}' = {p_{2}} + \delta$ for $\delta = 2e^{-t \varepsilon^{2}/32}$. 
For our choice of $c = \frac{\ln(1/(b_2 - \varepsilon))}{\ln(1/(b_1 + \varepsilon))}$ we plug the family $\LSH_{H}$ into the lower bound of Lemma \ref{lem:owzsimple} 
and use that $O(D^{-1/4}) = O(\varepsilon)$ which follows from our constraint that $\varepsilon \geq 1/d$.
\begin{align*}
	\rho(\LSH_{H}) &\geq 1/c - O(D^{-1/4}) \\
				   &= \frac{\ln(1/(1 + \varepsilon/b_{1})) + \ln(1/b_{1})}{\ln(1/(1-\varepsilon/b_{2})) + \ln(1/b_2)} - O(\varepsilon)\\
				   &\geq \frac{\ln(1/b_{1}) - \varepsilon/b_{1}}{\ln(1/b_2) + 2\varepsilon/b_{2}} - O(\varepsilon)\\
				   &= \frac{\ln(1/b_{1})}{\ln(1/b_2)} - O(\varepsilon) 
\end{align*}
Under our assumed properties of $\LSH_{B}$, we can upper bound the value of $\rho$ for $\LSH_{H}$. 
For simplicity we temporarily define $\lambda = 2\delta/{p_{2}}$ and assume that $\lambda / \ln(1/{p_{2}}) \leq 1/2$ and $\ln(1/{p_{2}}) \geq 1$.
The latter property holds without loss of generality through use of the standard LSH powering technique \cite{indyk1998, har-peled2012, odonnell2014optimal} 
that allows us to transform an LSH family with ${p_{2}} < 1$ to a family that has ${p_{2}} \leq 1/e$ without changing its associated $\rho$-value. 
\begin{align*}
	\rho(\LSH_{H}) &= \frac{\ln(1/{p_{1}}')}{\ln(1/{p_{2}}')} = \frac{ \ln(1/{p_{1}}) + \ln(1/(1-\delta/{p_{1}}))}{\ln(1/{p_{2}}) + \ln(1/(1 + \delta/{p_{2}}))} \\
				   &\leq \frac{\ln(1/{p_{1}}) + \lambda}{\ln(1/{p_{2}}) - \lambda} = \frac{\ln(1/{p_{1}}) + \lambda}{(\ln 1/{p_{2}} )(1 - \lambda / (\ln 1/{p_{2}}))} \\
				   &\leq \frac{\ln(1/{p_{1}}) + \lambda}{\ln(1/{p_{2}})}(1 +  2 \lambda / (\ln 1/{p_{2}})) = \frac{\ln(1/{p_{1}})}{\ln(1/{p_{2}})} + O(\delta / {p_{2}}) \\
				   &\leq \frac{\ln(1/b_{1})}{\ln(1/b_{2})} - \gamma  + O(\delta / {p_{2}}). 
\end{align*}
We get a contradiction between our upper bound and lower bound for $\rho(\LSH_{H})$ whenever $\gamma$ violates the following relation that summarizes the bounds:
\begin{equation*}
\frac{\ln(1/b_{1})}{\ln(1/b_2)} - O(\varepsilon) \leq  \rho(\LSH_{H}) \leq \frac{\ln(1/b_{1})}{\ln(1/b_{2})} - \gamma  + O(\delta / {p_{2}}).
\end{equation*}
In order for a contradiction to occur, the value of $\gamma$ has to satisfy
\begin{equation*}
\gamma > O(\varepsilon) + O(\delta / {p_{2}}).
\end{equation*}
By our setting of $t = \lfloor d/l \rfloor$ and $l = \lceil 8/\varepsilon \rceil$ we have that $\delta = e^{-\Omega(d\varepsilon^{3})}$.
We can cause a contradiction for a setting of $\varepsilon^{3} = K\frac{\ln(d/{p_{2}})}{d}$ where $K$ is some constant and where we assume that $d$ is greater than some constant.
The value of $\gamma$ for which the lower bound holds can be upper bounded by
\begin{equation*}
	\gamma = O\left(\frac{\ln(d/{p_{2}})}{d}\right)^{1/3}.
\end{equation*}
This completes the proof of Theorem \ref{thm:lower}.

\section{Equivalent Set Similarity Problems}\label{sec:equivalence}
%
In this section we consider how to use our data structure for Braun-Blanquet similarity search to support other similarity measures such as Jaccard similarity.
We already observed in the introduction that a direct translation exists between several similarity measures whenever the size of every sets is fixed to $t$.
Call an $(s_1,s_2)$-$S$-similarity search problem \emph{($t$,$t'$)-regular} if $P$ is restricted to vectors of weight $t$ and queries are restricted to vectors of weight~$t'$.
Obviously, a $(t,t')$-regular similarity search problem is no harder than the general similarity search problem, but it also cannot be too much easier when expressed as a function of the thresholds $(s_1,s_2)$:
For every pair $(t,t') \in \{0,\dots,d\}^2$ we can construct a ($t$,$t'$)-regular data structure (such that each point $\x \in P$ is represented in the $d+1$ data structures with $t=|\x|$), 
and answer a query for $\q\in\{0,1\}^d$ by querying all data structures with $t'=|\q|$.
Thus, the time and space for the general $(s_1,s_2)$-$S$-similarity search problem is at most $d+1$ times larger than the time and space of the most expensive ($t$,$t'$)-regular data structure.
This does \emph{not} mean that we cannot get better bounds in terms of other parameters, and in particular we expect that $(t,t')$-regular similarity search problems have difficulty that depends on parameters $t$ and $t'$.

\medskip

\para{Dimension Reduction.}
If the dimension is large a factor of $d$ may be significant.
However, for most natural similarity measures a $(s_1,s_2)$-$S$-similarity problem in $d\gg (\log n)^3$ dimensions can be reduced to a logarithmic number of $(s'_1,s'_2)$-$S$-similarity problems 
on $P'\subseteq \{0,1\}^{d'}$ in $d'=(\log n)^3$ dimensions with $s'_1 = s_1-O(1/\log n)$ and $s'_2 = s_2+O(1/\log n)$.
Since the similarity gap is close to the one in the original problem, $s'_1 - s'_2 = s_1 - s_2 - O(1/\log n)$, where $s_1$ and $s_2$ are assumed to be independent of $n$, the difficulty ($\rho$-value) remains essentially the same.
First, split $P$ into $\log d$ size classes $P_i$ such that vectors in class $i$ have size in $[2^i;2^{i+1})$.
For each size class the reduction is done independently and works by a standard technique: 
sample a sequence of random sets $I_j\subseteq \{1,\dots,d\}$, $i=1,\dots,d'$, and set $\x'_j = \vee_{\ell\in I_j} \x_\ell$.
The size of each set $I_j$ is chosen such that $Pr[\x'_j = 1] \approx 1/\log(n)$ when $|\x|=2^{i+1}$.
By Chernoff bounds this mapping preserves the relative weight of vectors up to size $2^i \log n$ up to an additive $O(1/\log n)$ term with high probability.
Assume now that the similarity measure $S$ is such that for vectors in $P_i$ we only need to consider $|\q|$ in the range from $2^i/\log n$ to $2^i \log n$ (since if the size difference is larger, the similarity is negligible).
The we can apply Chernoff bounds to the relative weights of the dimension-reduced vectors $\x'$, $\q'$ and the intersection $\x' \cap \q'$.
In particular, we get that the Jaccard similarity of a pair of vectors is preserved up to an additive error of $O(1/\log n)$ with high probability.
The class of similarity measures for which dimension reduction to $(\log n)^{O(1)}$ dimensions is possible is large, and we do not attempt to characterize it here.
Instead, we just note that for such similarity measures we can determine the complexity of similarity search up to a factor $(\log n)^{O(1)}$ by only considering regular search problems.

\medskip

\para{Equivalence of Regular Similarity Search Problems.}
We call a set similarity measure on $\{0,1\}^d$ \emph{symmetric} if it can be written in the form $S(\q,\x) = f_{d,|\q|,|\x|}(|\q\cap\x|)$, 
where each function $f_{d,|\q|,|\x|} \colon \mathbb{N} \rightarrow [0;1]$ is nondecreasing.
All 59 set similarity measures listed in the survey~\cite{choi2010survey}, normalized to yield similarities in $[0;1]$, are symmetric.
In particular this is the case for Jaccard similarity (where $J(\q,\x) = |\q\cap\x| /(|\q|+|\x|-|\q\cap\x|)$) and for Braun-Blanquet similarity.
For a symmetric similarity measure $S$, the predicate $S(\q,\x)\geq s_1$ is equivalent to the predicate $|\q\cap\x|\geq i_1$, where $i_1 = \min \{ i \; | \; f_{d,t',t}(i)\geq s_1 \}$, 
and $S(\q,\x) > s_2$ is equivalent to the predicate $|\q\cap\x|\geq i_2$, where $i_2 = \min \{ i \; | \; f_{d,t',t}(i) > s_2 \}$.
This means that every ($t$,$t'$)-regular $(s_1,s_2)$-$S$-similarity search problem on $P\subseteq \{0,1\}^d$ is equivalent to an $(i_1/d,i_2/d)$-$I$-similarity search problem on $P$, where $I(\q,\x)=|\x\cap\q|/d$.
In other words, all symmetric similarity search problems can be translated to each other, and it suffices to study a single one, such as Braun-Blanquet similarity.

\medskip

\para{Jaccard similarity.}
We briefly discuss Jaccard similarity since it is the most widely used measure of set similarity.
If we consider the problem of $(j_1, j_2)$-approximate Jaccard similarity search in the 
$(t,t')$-regular case with $t\ne t'$
then our Theorem \ref{thm:upper} is no longer guaranteed to yield the lowest value of $\rho$ among competing data-independent approaches such as MinHash and Angular LSH.
To simplify the comparision between different measures we introduce parameters $\beta$ and $b$ defined by $|\y| = \beta |\x|$ and $b = |\x \cap \y|/|\x|$ (note that $0 \leq b \leq \beta \leq 1$).
The three primary measures of set similarity considered in this paper can then be written as follows:
\begin{align*}
	B(\x, \y) &= b \\
	J(\x, \y) &= \frac{b}{1 + \beta - b} \\
	C(\x, \y) &= \frac{b}{\sqrt{\beta}}
\end{align*}
As shown in Figure \ref{fig:jaccard} among angular LSH, MinHash, and \textsc{Chosen Path}, the technique with the lowest $\rho$-value is different depending on the parameters $(j_1, j_2)$ and asymmetry $\beta$. 
\begin{figure*}
\subfloat[$\beta = 0.25$]{\includegraphics[width = 0.66\columnwidth]{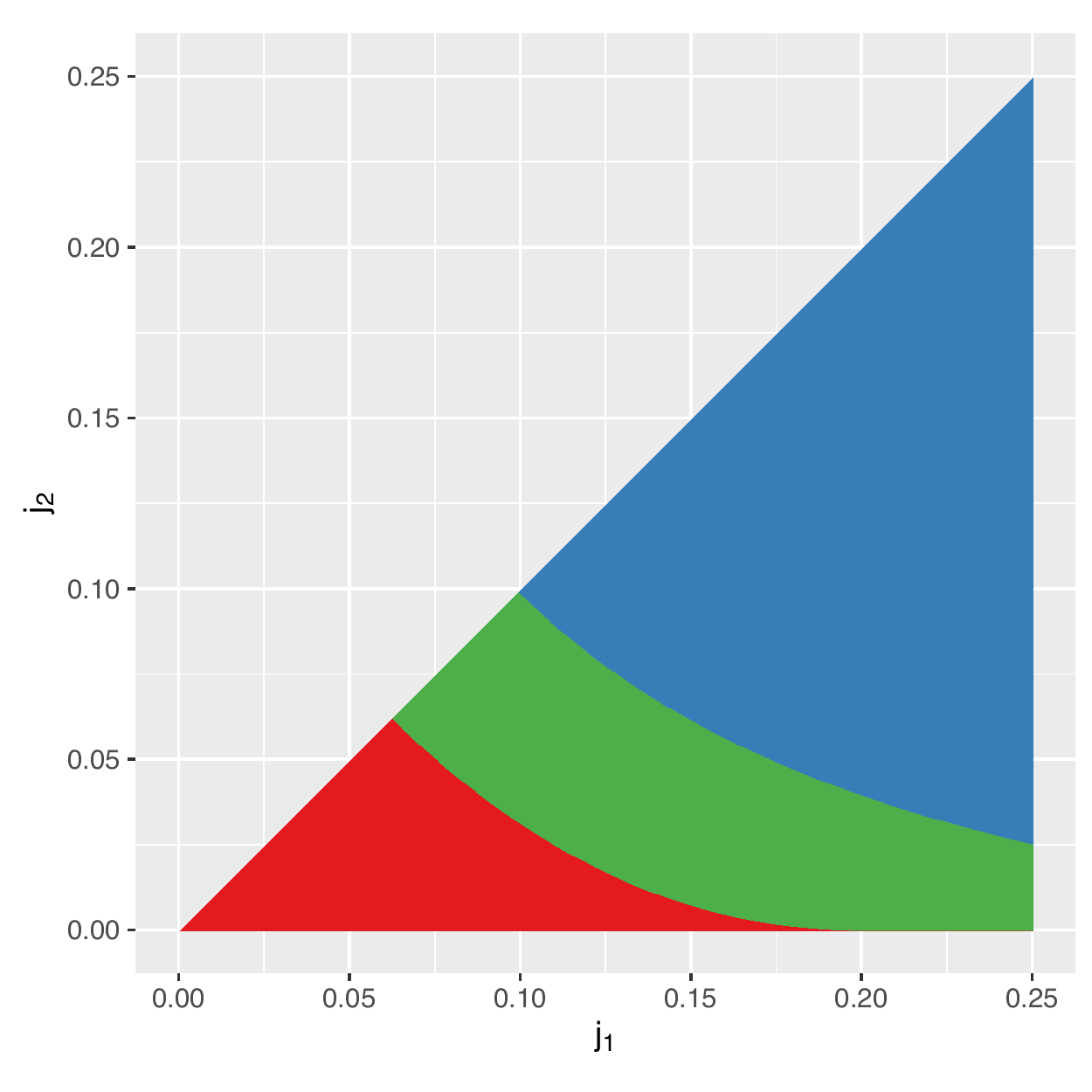}} 
\subfloat[$\beta = 0.5$]{\includegraphics[width = 0.66\columnwidth]{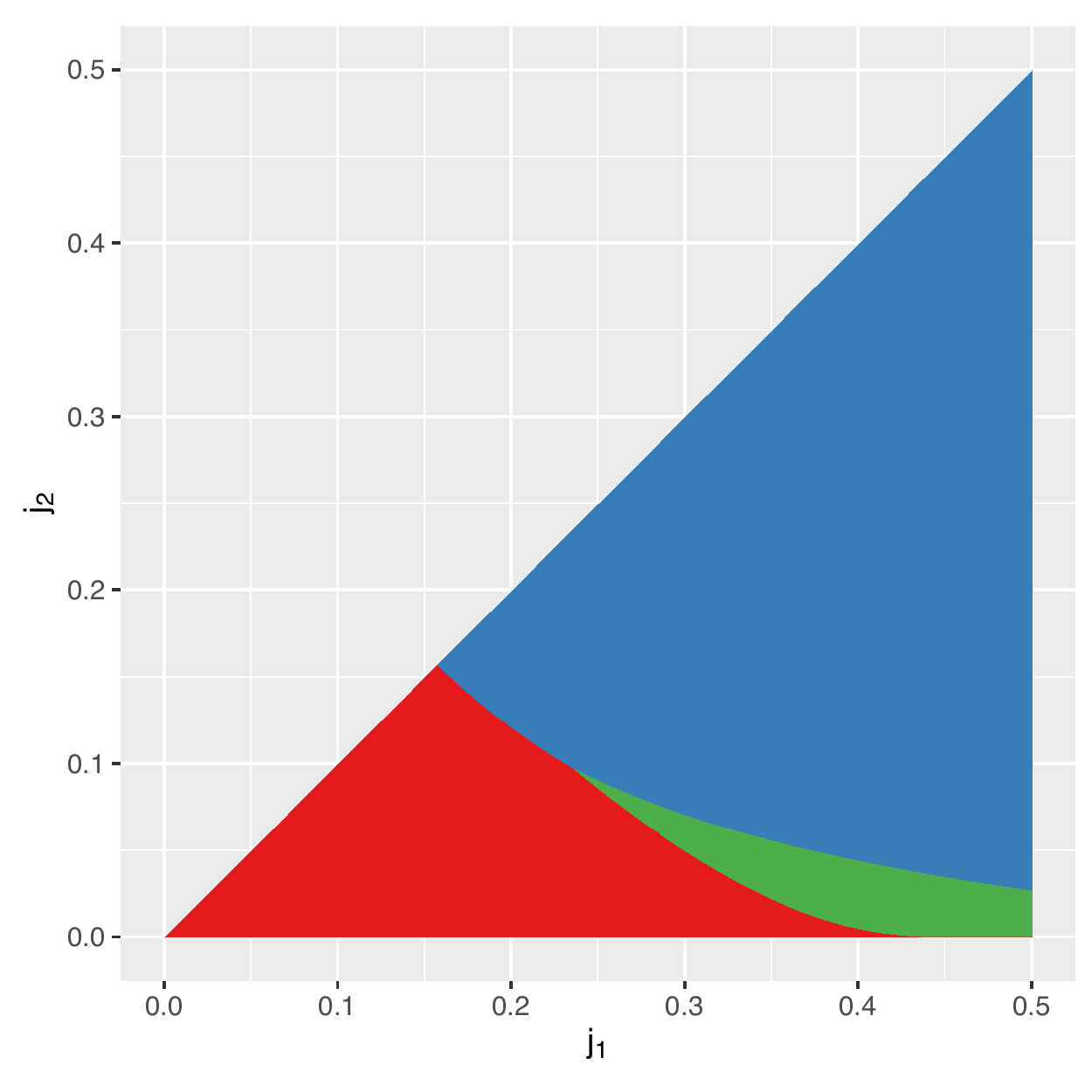}}    
\subfloat[$\beta = 0.75$]{\includegraphics[width = 0.66\columnwidth]{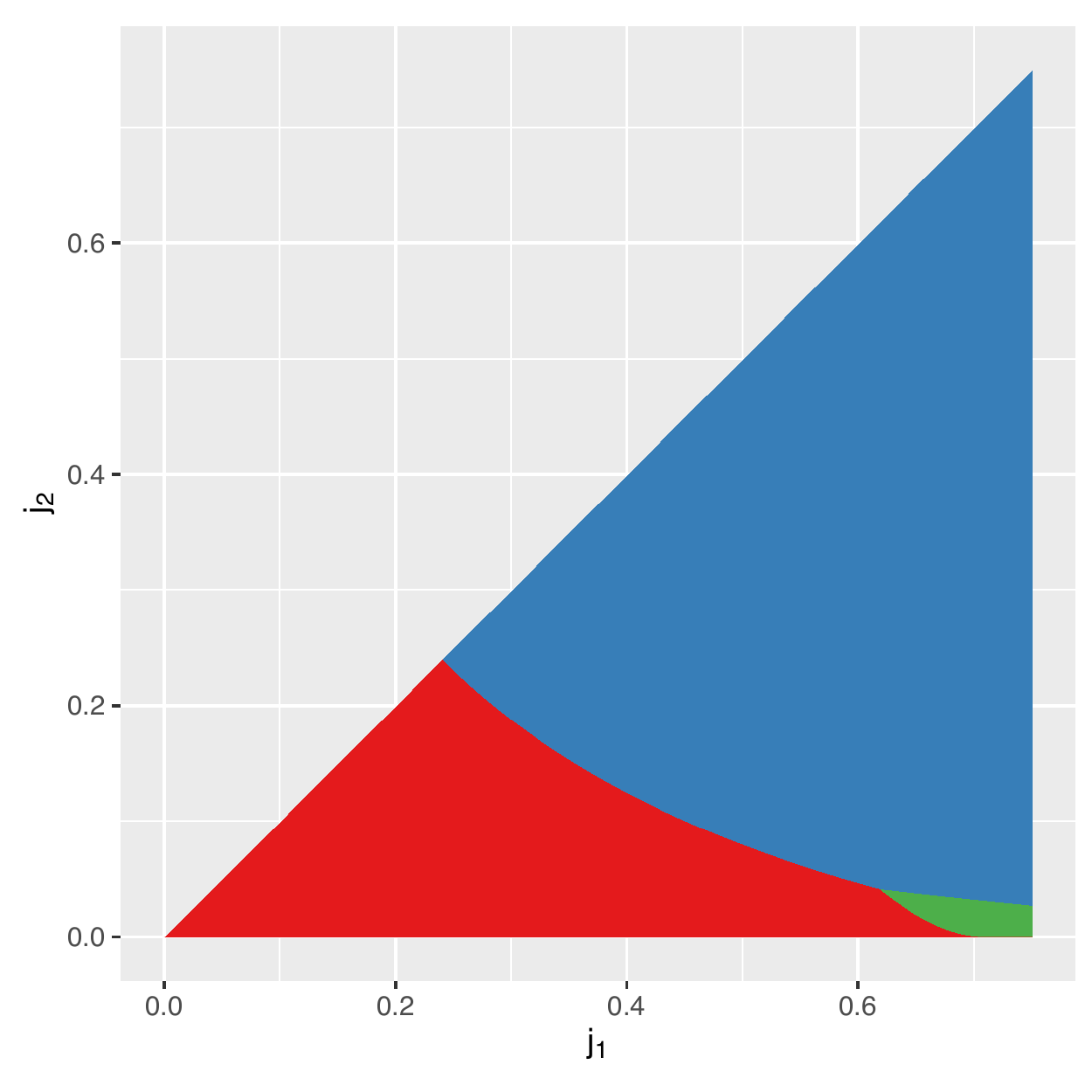}}    
\caption{Solution with lowest $\rho$-value for the $(j_1, j_2)$-approximate Jaccard similarity search problem for different values of $\beta$. 
Blue is angular LSH. Green is MinHash. Red is \textsc{Chosen Path}. 
Note the difference in the axes for different values of $\beta$ as it must hold that $0 \leq j_2 \leq j_1 \leq \beta$.}
\label{fig:jaccard}
\end{figure*}
We know that \textsc{Chosen Path} is optimal and strictly better than the competing data-independent techniques across the entire parameter space $(j_1, j_2)$ when $\beta = 1$, 
but it remains open to find tight upper and lower bounds in the case where $\beta \neq 1$.

\section{Conclusion and Open Problems}
We have seen that, perhaps surprisingly, there exists a relatively simple way of strictly improving the $\rho$-value for data-independent set similarity search in the case where all sets have the same size.
To implement the required locality-sensitive map efficiently we introduce a new technique based on branching processes 
that could possibly lead to more efficient solutions in other settings.

It remains an open problem to find tight upper and lower bounds on the $\rho$-value for Jaccard and cosine similarity search that hold for the entire parameter space in the general setting with arbitrary set sizes. 
Perhaps a modified version of the \textsc{Chosen Path} algorithm can yield an improved solution to Jaccard similarity search in general.
One approach is to generalize the condition $h_i(p \circ j) < \x_{j} / b_{1}|\x|$ to use different thresholds for queries and updates.
This yields different space-time tradeoffs when applying the \textsc{Chosen Path} algorithm to Jaccard similarity search.

Another interesting question is if the improvement shown for sparse vectors can be achieved in general for inner product similarity.
A similar, but possibly easier, direction would be to consider \emph{weighted} Jaccard similarity.

\begin{acks}
We thank Thomas Dybdahl Ahle for comments on a previous version of this manuscript.
\end{acks}

\appendix
\section{Details Behind the Lower Bound}\label{app:lower}
\subsection{Tools}
For clarity we state some standard technical lemmas that we use to derive LSH lower bounds. 
\begin{lemma}[{Hoeffding \cite[Theorem 1]{hoeffding1963}}] \label{lem:hoeffding}
	Let $X_1, X_2, \dots, X_n$ be independent random variables satisfying $0 \leq X_i \leq 1$ for $i \in [n]$.
	Define $X = X_1 + X_2 + \dots + X_n$, $Z = X/n$, and $\mu = \E[Z]$, then:
	\begin{itemize}
		\item[-] For $\hat{\mu} \geq \mu$ and $0 < \varepsilon < 1 - \hat{\mu}$ we have that $\Pr[Z - \hat{\mu} \geq \varepsilon] \leq e^{-2n\varepsilon^{2}}$.
		\item[-] For $\hat{\mu} \leq \mu$ and $0 < \varepsilon < \hat{\mu}$ we have that $\Pr[Z - \hat{\mu} \leq - \varepsilon] \leq e^{-2n\varepsilon^{2}}$.
	\end{itemize}
\end{lemma}
\begin{lemma}[{Chernoff \cite[Thm.~4.4 and~4.5]{mitzenmacher2005}}] \label{lem:chernoff} 
 	Let $X_1, \dots, X_n$ be independent Poisson trials and define $X = \sum_{i=1}^{n}X_i$ and $\mu = \E[X]$.
 	Then, for $0 < \varepsilon < 1$ we have 
 	\begin{itemize}
 		\item[-] $\Pr[X \geq (1+\varepsilon)\mu] \leq e^{-\varepsilon^2 \mu / 3}$.
 		\item[-] $\Pr[X \leq (1-\varepsilon)\mu] \leq e^{-\varepsilon^2 \mu / 2}$.
 	\end{itemize}
\end{lemma}
\begin{lemma}[Bounding the logarithm {\cite{topsoe2007}}] \label{lem:lnbounds}
 	For $x > -1$ we have that $\tfrac{x}{1+x} \leq \ln(1 + x) \leq x$. 
\end{lemma}
\begin{lemma}[Approximating the exponential function {\cite[Prop.~B.3]{motwani2010randomized}}]\label{lem:exp} 
 	For all $t, n \in \real$ with $|t| \leq n$ we have that $e^{t}(1 - \tfrac{t^2}{n}) \leq (1 + \tfrac{t}{n})^n \leq e^t$.
\end{lemma}
\subsection{Proof of Lemma \ref{lem:owzsimple}} 
\para{Preliminaries.}
We will reuse the notation of Section 3.~from O'Donnell et al.~\cite{odonnell2014optimal}.
\begin{definition}
	For $0 \leq \lambda < 1$ we say that $(\x, \y)$ are $(1-\lambda)$-correlated if $\x$ is chosen uniformly at random from $\cube{d}$ 
and $\y$ is constructed by rerandomizing each bit from $\x$ independently at random with probability $\lambda$.
\end{definition}
Let $(\x, \y)$ be $e^{-t}$-correlated and let $\LSH$ be a family of hash functions on $\cube{d}$, then we define
\begin{equation*}
	\K_{\LSH}(t) = \Pr_{\substack{h \sim \LSH \\ (\x, \y)\, e^{-t}\text{- corr'd}}}[h(\x) = h(\y)].
\end{equation*}
We have that $\K_{\LSH}(t)$ is a log-convex function which implies the following property that underlies the lower bound: 
\begin{lemma}\label{lem:logconvexity}
	For every family of hash functions $\LSH$ on $\cube{d}$, every $t \geq 0$, and $c \geq 1$ we have
	\begin{equation}
		\frac{\ln(1/\K_{\LSH}(t))}{\ln(1/\K_{\LSH}(ct))} \geq \frac{1}{c}.
	\end{equation}
\end{lemma}
The idea behind the proof is to tie $p_1$ to $\sen{t}$ and $p_2$ to $\sen{ct}$ through Chernoff bounds and then apply Lemma \ref{lem:logconvexity} to show that $\rho \gtrsim 1/c$.

\para{Proof.}
Begin by assuming that we have a family $\LSH$ that satisfies the conditions of Lemma \ref{lem:owzsimple}.
Note that the expected Hamming distance betwee $(1-\lambda)$-correlated points $\x$ and $\y$ is given by $(\lambda/2)d$.
We set $\lambda_{p_{1}}/2 = d^{-1/2} - d^{-5/8}$ and $\lambda_{p_{2}}/2 = cd^{-1/2} + 2cd^{-5/8}$ and let $(\x, \y)$ denote $(1 - \lambda_{p_{1}})$-correlated random strings 
and $(\x, \x')$ denote $(1 - \lambda_{p_{2}}q$)-correlated random strings.
By standard Chernoff bounds we get the following guarantees:
\begin{align*}
	\Pr[\norm{\x-\y}_{1} \geq r] &\leq e^{-\Omega(d^{1/4})}, \\
	\Pr[\norm{\x-\x'}_1 \leq cr] &\leq e^{-\Omega(d^{1/4})}.
\end{align*}
We will establish a relationship between $\K_{\LSH}(t_{p_{1}})$ and ${p_{1}}$ on the one hand, and $\K_{\LSH}(t_{p_{2}})$ and ${p_{2}}$ on the other hand, 
for the following choice of parameters $t_{p_{1}}$ and $t_{p_{2}}$:
\begin{align*}
	t_{p_{1}} &= -\ln(1 - 2(d^{-1/2} - d^{-5/8})) \\
	t_{p_{2}} &= -\ln(1 - 2c(d^{-1/2} + 2d^{-5/8})).
\end{align*}
By the properties of $\LSH$ and from the definition of $\K_{\LSH}$ we have that
\begin{align*}
	\K_{\LSH}(t_{p_{1}}) &\geq {p_{1}}(1 - \Pr[\norm{\x - \y}_1 > r]) \geq {p_{1}} - \Pr[\norm{\x - \y}_1 \geq r] \\
	\K_{\LSH}(t_{p_{2}}) &\leq {p_{2}}(1 - \Pr[\norm{\x - \x'}_1 \leq cr]) + \Pr[\norm{\x - \x'}_1 \leq cr] \\
						 &\leq {p_{2}} + \Pr[\norm{\x - \x'}_1 \leq cr].
\end{align*}
Let $\delta = \max\{\Pr[\norm{\x - \y}_1 \geq r], \Pr[\norm{\x - \x'}_1 \leq cr]\} = e^{-\Omega(d^{1/4})}$.
By Lemma \ref{lem:logconvexity} and our setting of $t_{p_{1}}$ and $t_{p_{2}}$ we can use the bounds on the natural logarithm from Lemma \ref{lem:lnbounds} to show the following: 
\begin{align*}
	\frac{\ln(1/\K_{\LSH}(t_{p_{1}}))}{\ln(1/\K_{\LSH}(t_{p_{2}}))} &\geq \frac{t_{p_{1}}}{t_{p_{2}}} = \frac{\ln(1 - 2(d^{-1/2} - d^{-5/8}))}{\ln(1 - 2c(d^{-1/2} + 2d^{-5/8}))} \\
	&\geq \frac{2(d^{-1/2} - d^{-5/8})}{2c(d^{-1/2} + 2d^{-5/8})} -  2(d^{-1/2} - d^{-5/8}) \\
	&\geq \frac{1 - d^{-1/4}}{c + 2d^{-1/4}} - 2(d^{-1/2} - d^{-5/8}) \\
	&= \frac{1}{c} - O(d^{-1/4}).
\end{align*}
We proceed by lower bounding $\rho$ where we make use of the inequalities derived above.
\begin{equation*}
\sen{t_{p_{2}}} - \delta \leq {p_{2}} < {p_{1}} \leq \sen{t_{p_{1}}} + \delta.
\end{equation*}
By Lemma \ref{lem:logconvexity} combined with the restrictions on our parameters,
for $d$ greater than some constant we have that $\sen{t_{p_{2}}} \geq \sen{t_{p_{1}}}^{2c} \geq ({p_{1}}/2)^{2c} \geq (2d)^{-2c} \geq (2d)^{-2d^{1/8}}$.
Furthermore, we lower bound $\ln(1/\sen{t_{p_{2}}})$ by using that $\sen{t_{p_{2}}} \leq {p_{2}} + \delta$ together with the restriction that ${p_{2}} \geq 1 - 1/d$ and the properties of $\delta$.
For $d$ greater than some constant it therefore holds that $\sen{t_{p_{2}}} \leq 1 - 1/2d$ from which it follows that $\ln(1/\sen{t_{p_{2}}}) \geq 1/2d$.   
\begin{align*}
	\frac{\ln(1/{p_{1}})}{\ln(1/{p_{2}})} &\geq \frac{\ln(1/(\sen{t_{p_{1}}} + \delta))}{\ln(1/(\sen{t_{p_{2}}} - \delta))} \\
							  &= \frac{\ln(1/\sen{t_{p_{1}}}) - \ln(1 + \delta/\sen{t_{p_{1}}})}{\ln(1/\sen{t_{p_{2}}}) + \ln(1/(1 - \delta/\sen{t_{p_{2}}}))} \\
							  &\geq \frac{\ln(1/\sen{t_{p_{1}}}) - \delta/\sen{t_{p_{1}}}}{\ln(1/\sen{t_{p_{2}}}) + 2\delta/\sen{t_{p_{2}}}} \\
							  &\geq \frac{\ln(1/\sen{t_{p_{1}}})}{\ln(1/\sen{t_{p_{2}}})} - \frac{3 \delta}{\sen{t_{p_{2}}} \ln(1/\sen{t_{p_{2}}})}.
\end{align*}
By the arguments above we have that 
\begin{equation*}
	\frac{3 \delta}{\sen{t_{p_{2}}} \ln(1/\sen{t_{p_{2}}})} = e^{-\Omega(d^{1/4})} = O(d^{-1/4}).
\end{equation*}
Inserting the lower bound for $\frac{\ln(1/\sen{t_{p_{1}}})}{\ln(1/\sen{t_{p_{2}}})}$ results in the lemma. 

\section{Comparisons} \label{app:comparison}
For completeness we state the proofs behind the comparisons between the $\rho$-values obtained by the \textsc{Chosen Path} algorithm and other LSH techniques.
\subsection{MinHash}
For data sets with fixed sparsity and Braun-Blanquet similarities $0 < b_2 < b_1 < 1$ we have that $\rho/\rho_{\text{minhash}} = f(b_{2})/f(b_{1})$ where $f(x) = \log(x/(2-x)) / \log(x)$.
If $f(x)$ is monotone increasing in $(0;1)$ then $\rho/\rho_{\text{minhash}} < 1$.
For $x \in (0;1)$ we have that $\sign(f'(x)) = \sign(g(x))$ where $g(x) = \ln(x) + (2-x) \ln(2-x)$.
The function $g(x)$ equals zero at $x = 1$ and has the derivative $g'(x) = \ln(x) - \ln(2-x)$ which is negative for values of $x \in (0;1)$.
We can thefore see that $f'(x)$ is positive in the interval and it follows that $\rho < \rho_{\text{minhash}}$ for every choice of $0 < b_2 < b_1 < 1$. 

\subsection{Angular LSH}
We have that $\rho/\rho_{\text{angular}} < 1$ if $f(x) = \ln(x) \frac{1+x}{1-x}$ is a monotone increasing function for $x \in (0;1)$.
For $x \in (0;1)$ we have that $\sign(f'(x)) = \sign(g(x))$ where $g(x) = (1-x^2)/2 + x \ln x$.
We note that $g(1) = 0$ and $g'(x) = 1 - x + \ln x$. 
Therefore, if $g'(x) < 0$ for $x \in (0;1)$ it holds that $g(x) > 0$ and $f(x)$ is monotone increasing in the same interval.
We have that $g'(1) = 0$ and $g''(x) = -1 + 1/x > 0$ implying that $g'(x) < 0$ in the interval. 

\subsection{Data-dependent LSH}

\begin{lemma}\label{lem:fixedrho}
Let $0 < b_2 < b_1 < 1$ and fix $\rho = 1/2$ such that $b_1 = \sqrt{b_2}$.
Then we have that $\rho < \rho_{\text{datadep}}$ for every value of $b_2 < 1/4$. 
\end{lemma}
\begin{proof}
We will compare $\rho = \log(b_{1})/\log(b_{2})$ and $\rho_{\text{datadep}} = \frac{1 - b_{1}}{1 + b_{1} - 2b_{2}}$ when $\rho$ is fixed at $\rho = 1/2$, or equivalently, $b_{1} = \sqrt{b_{2}}$.
We can solve the quadratic equation $1/2 = \frac{1 - \sqrt{b_2}}{1 + \sqrt{b_2} - 2b_{2}}$ to see that for $0 < b_2 < 1$ we have that $\rho = \rho_{\text{datadep}}$ only when $b_2 = 1/4$. 
The derivative of $\rho_{\text{datadep}}$ with respect to $b_2$ is negative when $b_{1} = \sqrt{b_2}$. 
Under this restriction we therefore have that $\rho < \rho_{\text{datadep}}$ for $b_2 < 1/4$ which is equivalent to $j_2 < 1/7$ in the fixed-weight setting.
\end{proof}

To compare $\rho$-values over the full parameter space we use the following two lemmas. 
\begin{lemma} \label{lem:partial}
For every choice of fixed $0 < \rho < 1$ let $b_{2} = b_{1}^{1/\rho}$. 
Then $\rho_{\text{datadep}} = \frac{1 - b_{1}}{1 + b_{1} - 2b_{2}}$ is decreasing in $b_1$ for $b_1 \in (0;1)$.
\end{lemma}
\begin{proof}
The sign of the derivative of $\rho_{\text{datadep}}$ with respect to $b_1$ is equal to the sign of the function $g(x) = -\rho x^{-1/\rho} + \rho - 1 + x^{-1}$ for $x \in (0;1)$.
We have that $g(1) = 0$ and $g'(x) = x{-1/p - 1} - x^{-2} > 0$ for $x \in (0;1)$ which shows that $g(x) < 0$ in the interval.
\end{proof}

\begin{lemma} \label{lem:onefifth}
For $1/5 = b_2 < b_1 < 1$ we have that $\rho < \rho_{\text{datadep}}$.
\end{lemma}
\begin{proof}
For fixed $b_2 = 1/5$ consider $f(b_1) = \rho - \rho_{\text{datadep}}$ as a function of $b_1$ in the interval $[1/5, 1]$.
We want to show that $f(b_1) < 0$ for $b_1 \in (1/5;1)$. 
In the endpoints the function takes the value $0$. 
Between the endpoints we find that $f'(b_1) = \frac{1}{\ln(5)b_1} + \frac{8/5}{(3/5 + b_1)^2}$ and that $f'(b_1) = 0$ is a quadratic form with only one solution $b_{1}^{*}$ in $[1/5;1]$.
By Lemma \ref{lem:fixedrho} we know that that for $b_2 = 1/5$ and $b_1 = 1/\sqrt{5}$ it holds that $f(b_1) < 0$.
Since $f(1/5) = f(1) = 0$, $f'(b_1) = 0$ only in a single point in $[1/5;1]$, and $f(1/\sqrt{5}) < 0$ we can conclude that the lemma holds.
\end{proof}

\begin{corollary}
For every choice of $b_1, b_2$ satisfying $0 < b_2 \leq 1/5$ and $b_2 < b_1 < 1$ we have that $\rho < \rho_{\text{datadep}}$.
\end{corollary}
\begin{proof}
If $b_2 = 1/5$ the property holds by Lemma \ref{lem:onefifth}.
If $b_2 < 1/5$ we define new variables $\hat{b}_2, \hat{b}_2$, setting $\hat{b}_1 = \hat{b}_{1}^{\rho(b_1, b_2)}$ and initially consider $\hat{b}_{2} = 1/5$.
In this setting we again have that $\rho(\hat{b}_{1}, \hat{b}_{2}) < \rho_{\text{datadep}}(\hat{b}_{1},\hat{b}_{2})$.
According to Lemma \ref{lem:partial} it holds that $\rho_{\text{datadep}}$ is decreasing in $b_2$ for fixed $\rho$.
Therefore, as $\hat{b}_{2}$ decreases to $\hat{b}_{2} = b_2$ where $\hat{b}_{1} = b_1$ we have that $\rho(\hat{b}_{1}, \hat{b}_{2}) = \rho$ remains constant while $\rho_{\text{datadep}}$ increases. 
Since it held that $\rho < \rho_{\text{datadep}}$ at the initial values of $\hat{b}_{1}, \hat{b}_{2}$ it must also hold for $b_1, b_2$.
\end{proof}

\para{Numerical Comparison of MinHash and Data-dep.~LSH.}
Comparing $\rho_{\text{minhash}}$ to $\rho_{\text{datadep}}$ we can verify numerically that even for $b_2$ fixed as low as $b_2 = 1/23$ 
we can find values of $b_1$ (for example $b_1 = 0.995$ such that $\rho_{\text{minhash}} > \rho_{\text{datadep}}$.

\bibliographystyle{ACM-Reference-Format}
\bibliography{sparse}


\begin{thebibliography}{00}


\ifx \showCODEN    \undefined \def \showCODEN     #1{\unskip}     \fi
\ifx \showDOI      \undefined \def \showDOI       #1{{\tt DOI:}\penalty0{#1}\ }
  \fi
\ifx \showISBNx    \undefined \def \showISBNx     #1{\unskip}     \fi
\ifx \showISBNxiii \undefined \def \showISBNxiii  #1{\unskip}     \fi
\ifx \showISSN     \undefined \def \showISSN      #1{\unskip}     \fi
\ifx \showLCCN     \undefined \def \showLCCN      #1{\unskip}     \fi
\ifx \shownote     \undefined \def \shownote      #1{#1}          \fi
\ifx \showarticletitle \undefined \def \showarticletitle #1{#1}   \fi
\ifx \showURL      \undefined \def \showURL       {\relax}        \fi
\providecommand\bibfield[2]{#2}
\providecommand\bibinfo[2]{#2}
\providecommand\natexlab[1]{#1}
\providecommand\showeprint[2][]{arXiv:#2}

\bibitem[\protect\citeauthoryear{Ahle, Pagh, Razenshteyn, and Silvestri}{Ahle
  et~al\mbox{.}}{2016}]%
        {ahle2015}
\bibfield{author}{\bibinfo{person}{T.~D. Ahle}, \bibinfo{person}{R. Pagh},
  \bibinfo{person}{I.~P. Razenshteyn}, {and} \bibinfo{person}{F. Silvestri}.}
  \bibinfo{year}{2016}\natexlab{}.
\newblock \showarticletitle{On the Complexity of Inner Product Similarity
  Join}. In \bibinfo{booktitle}{{\em Proc. {PODS}'16}}.
  \bibinfo{pages}{151--164}.
\newblock


\bibitem[\protect\citeauthoryear{Andoni, Indyk, Laarhoven, Razenshteyn, and
  Schmidt}{Andoni et~al\mbox{.}}{2015}]%
        {andoni2015practical}
\bibfield{author}{\bibinfo{person}{A. Andoni}, \bibinfo{person}{P. Indyk},
  \bibinfo{person}{T. Laarhoven}, \bibinfo{person}{I. Razenshteyn}, {and}
  \bibinfo{person}{L. Schmidt}.} \bibinfo{year}{2015}\natexlab{}.
\newblock \showarticletitle{Practical and optimal LSH for angular distance}. In
  \bibinfo{booktitle}{{\em Proc. {NIPS} '15}}. \bibinfo{pages}{1225--1233}.
\newblock


\bibitem[\protect\citeauthoryear{Andoni, Indyk, Nguyen, and Razenshteyn}{Andoni
  et~al\mbox{.}}{2014}]%
        {andoni2014beyond}
\bibfield{author}{\bibinfo{person}{A. Andoni}, \bibinfo{person}{P. Indyk},
  \bibinfo{person}{H.~L. Nguyen}, {and} \bibinfo{person}{I.~P. Razenshteyn}.}
  \bibinfo{year}{2014}\natexlab{}.
\newblock \showarticletitle{Beyond Locality-Sensitive Hashing}. In
  \bibinfo{booktitle}{{\em Proc. {SODA} '14}}. \bibinfo{pages}{1018--1028}.
\newblock


\bibitem[\protect\citeauthoryear{Andoni, Laarhoven, Razenshteyn, and
  Waingarten}{Andoni et~al\mbox{.}}{2017}]%
        {andoni2017optimal}
\bibfield{author}{\bibinfo{person}{A. Andoni}, \bibinfo{person}{T. Laarhoven},
  \bibinfo{person}{I.~P. Razenshteyn}, {and} \bibinfo{person}{E. Waingarten}.}
  \bibinfo{year}{2017}\natexlab{}.
\newblock \showarticletitle{Optimal Hashing-based Time-Space Trade-offs for
  Approximate Near Neighbors}. In \bibinfo{booktitle}{{\em Proc. {SODA} '17}}.
  \bibinfo{pages}{47--66}.
\newblock


\bibitem[\protect\citeauthoryear{Andoni and Razenshteyn}{Andoni and
  Razenshteyn}{2015}]%
        {andoni2015optimal}
\bibfield{author}{\bibinfo{person}{A. Andoni} {and} \bibinfo{person}{I.
  Razenshteyn}.} \bibinfo{year}{2015}\natexlab{}.
\newblock \showarticletitle{Optimal Data-Dependent Hashing for Approximate Near
  Neighbors}. In \bibinfo{booktitle}{{\em Proc. {STOC} '15}}.
  \bibinfo{pages}{793--801}.
\newblock


\bibitem[\protect\citeauthoryear{Andoni and Razensteyn}{Andoni and
  Razensteyn}{2016}]%
        {andoni2016tight}
\bibfield{author}{\bibinfo{person}{A. Andoni} {and} \bibinfo{person}{I.
  Razensteyn}.} \bibinfo{year}{2016}\natexlab{}.
\newblock \showarticletitle{Tight Lower Bounds for Data-Dependent
  Locality-Sensitive Hashing}. In \bibinfo{booktitle}{{\em Proc. {SoCG} '16}}.
  \bibinfo{pages}{9:1--9:11}.
\newblock


\bibitem[\protect\citeauthoryear{Arasu, Ganti, and Kaushik}{Arasu
  et~al\mbox{.}}{2006}]%
        {arasu2006efficient}
\bibfield{author}{\bibinfo{person}{Arvind Arasu}, \bibinfo{person}{Venkatesh
  Ganti}, {and} \bibinfo{person}{Raghav Kaushik}.}
  \bibinfo{year}{2006}\natexlab{}.
\newblock \showarticletitle{Efficient exact set-similarity joins}. In
  \bibinfo{booktitle}{{\em Proceedings of the 32nd international conference on
  Very large data bases}}. VLDB Endowment, \bibinfo{pages}{918--929}.
\newblock


\bibitem[\protect\citeauthoryear{Bayardo, Ma, and Srikant}{Bayardo
  et~al\mbox{.}}{2007}]%
        {bayardo2007scaling}
\bibfield{author}{\bibinfo{person}{Roberto~J Bayardo}, \bibinfo{person}{Yiming
  Ma}, {and} \bibinfo{person}{Ramakrishnan Srikant}.}
  \bibinfo{year}{2007}\natexlab{}.
\newblock \showarticletitle{Scaling up all pairs similarity search}. In
  \bibinfo{booktitle}{{\em Proceedings of the 16th international conference on
  World Wide Web}}. ACM, \bibinfo{pages}{131--140}.
\newblock


\bibitem[\protect\citeauthoryear{Becker, Ducas, Gama, and Laarhoven}{Becker
  et~al\mbox{.}}{2016}]%
        {becker2016}
\bibfield{author}{\bibinfo{person}{A. Becker}, \bibinfo{person}{L. Ducas},
  \bibinfo{person}{N. Gama}, {and} \bibinfo{person}{T. Laarhoven}.}
  \bibinfo{year}{2016}\natexlab{}.
\newblock \showarticletitle{New directions in nearest neighbor searching with
  applications to lattice sieving}. In \bibinfo{booktitle}{{\em Proc. {SODA}
  '16}}. \bibinfo{pages}{10--24}.
\newblock


\bibitem[\protect\citeauthoryear{Braun-Blanquet}{Braun-Blanquet}{1932}]%
        {braunblanquet1932}
\bibfield{author}{\bibinfo{person}{Josias Braun-Blanquet}.}
  \bibinfo{year}{1932}\natexlab{}.
\newblock \bibinfo{booktitle}{{\em Plant sociology. The study of plant
  communities}}.
\newblock \bibinfo{publisher}{McGraw-Hill}.
\newblock


\bibitem[\protect\citeauthoryear{Broder}{Broder}{1997}]%
        {Bro97b}
\bibfield{author}{\bibinfo{person}{Andrei~Z. Broder}.}
  \bibinfo{year}{1997}\natexlab{}.
\newblock \showarticletitle{On the resemblance and containment of documents}.
  In \bibinfo{booktitle}{{\em Compression and Complexity of Sequences 1997.
  Proceedings}}. IEEE, \bibinfo{pages}{21--29}.
\newblock


\bibitem[\protect\citeauthoryear{Broder, Glassman, Manasse, and Zweig}{Broder
  et~al\mbox{.}}{1997}]%
        {Bro97a}
\bibfield{author}{\bibinfo{person}{Andrei~Z. Broder},
  \bibinfo{person}{Steven~C. Glassman}, \bibinfo{person}{Mark~S. Manasse},
  {and} \bibinfo{person}{Geoffrey Zweig}.} \bibinfo{year}{1997}\natexlab{}.
\newblock \showarticletitle{Syntactic clustering of the web}.
\newblock \bibinfo{journal}{{\em Computer Networks and ISDN Systems\/}}
  \bibinfo{volume}{29}, \bibinfo{number}{8} (\bibinfo{year}{1997}),
  \bibinfo{pages}{1157--1166}.
\newblock


\bibitem[\protect\citeauthoryear{Charikar}{Charikar}{2002}]%
        {charikar2002}
\bibfield{author}{\bibinfo{person}{M. Charikar}.}
  \bibinfo{year}{2002}\natexlab{}.
\newblock \showarticletitle{Similarity estimation techniques from rounding
  algorithms}. In \bibinfo{booktitle}{{\em Proc. {STOC} '02}}.
  \bibinfo{pages}{380--388}.
\newblock


\bibitem[\protect\citeauthoryear{Chierichetti and Kumar}{Chierichetti and
  Kumar}{2015}]%
        {chierichetti2015}
\bibfield{author}{\bibinfo{person}{F. Chierichetti} {and} \bibinfo{person}{R.
  Kumar}.} \bibinfo{year}{2015}\natexlab{}.
\newblock \showarticletitle{LSH-Preserving Functions and Their Applications}.
\newblock \bibinfo{journal}{{\em J. {ACM}\/}} \bibinfo{volume}{62},
  \bibinfo{number}{5} (\bibinfo{year}{2015}), \bibinfo{pages}{33}.
\newblock


\bibitem[\protect\citeauthoryear{Choi, Cha, and Tappert}{Choi
  et~al\mbox{.}}{2010}]%
        {choi2010survey}
\bibfield{author}{\bibinfo{person}{S. Choi}, \bibinfo{person}{S. Cha}, {and}
  \bibinfo{person}{C.~C. Tappert}.} \bibinfo{year}{2010}\natexlab{}.
\newblock \showarticletitle{A survey of binary similarity and distance
  measures}.
\newblock \bibinfo{journal}{{\em J. Syst. Cybern. Informatics\/}}
  \bibinfo{volume}{8}, \bibinfo{number}{1} (\bibinfo{year}{2010}),
  \bibinfo{pages}{43--48}.
\newblock


\bibitem[\protect\citeauthoryear{Christiani}{Christiani}{2017}]%
        {christiani2017framework}
\bibfield{author}{\bibinfo{person}{T. Christiani}.}
  \bibinfo{year}{2017}\natexlab{}.
\newblock \showarticletitle{A Framework for Similarity Search with Space-Time
  Tradeoffs using Locality-Sensitive Filtering}. In \bibinfo{booktitle}{{\em
  Proc. {SODA} '17}}. \bibinfo{pages}{31--46}.
\newblock


\bibitem[\protect\citeauthoryear{Cohen}{Cohen}{1997}]%
        {cohen1997size}
\bibfield{author}{\bibinfo{person}{E. Cohen}.} \bibinfo{year}{1997}\natexlab{}.
\newblock \showarticletitle{Size-estimation framework with applications to
  transitive closure and reachability}.
\newblock \bibinfo{journal}{{\em J. Comp. Syst. Sci.\/}} \bibinfo{volume}{55},
  \bibinfo{number}{3} (\bibinfo{year}{1997}), \bibinfo{pages}{441--453}.
\newblock


\bibitem[\protect\citeauthoryear{Cohen and Kaplan}{Cohen and Kaplan}{2009}]%
        {cohen2009leveraging}
\bibfield{author}{\bibinfo{person}{E. Cohen} {and} \bibinfo{person}{H.
  Kaplan}.} \bibinfo{year}{2009}\natexlab{}.
\newblock \showarticletitle{Leveraging discarded samples for tighter estimation
  of multiple-set aggregates}.
\newblock \bibinfo{journal}{{\em ACM SIGMETRICS Performance Evaluation
  Review\/}} \bibinfo{volume}{37}, \bibinfo{number}{1} (\bibinfo{year}{2009}),
  \bibinfo{pages}{251--262}.
\newblock


\bibitem[\protect\citeauthoryear{Dubiner}{Dubiner}{2010}]%
        {dubiner2010bucketing}
\bibfield{author}{\bibinfo{person}{M. Dubiner}.}
  \bibinfo{year}{2010}\natexlab{}.
\newblock \showarticletitle{Bucketing coding and information theory for the
  statistical high-dimensional nearest-neighbor problem}.
\newblock \bibinfo{journal}{{\em {IEEE} Trans. Information Theory\/}}
  \bibinfo{volume}{56}, \bibinfo{number}{8} (\bibinfo{year}{2010}),
  \bibinfo{pages}{4166--4179}.
\newblock


\bibitem[\protect\citeauthoryear{Hagerup}{Hagerup}{1998}]%
        {hagerup1998}
\bibfield{author}{\bibinfo{person}{T. Hagerup}.}
  \bibinfo{year}{1998}\natexlab{}.
\newblock \showarticletitle{Sorting and Searching on the Word {RAM}}. In
  \bibinfo{booktitle}{{\em Proc. {STACS} '98}}. \bibinfo{pages}{366--398}.
\newblock


\bibitem[\protect\citeauthoryear{Har-Peled, Indyk, and Motwani}{Har-Peled
  et~al\mbox{.}}{2012}]%
        {har-peled2012}
\bibfield{author}{\bibinfo{person}{S. Har-Peled}, \bibinfo{person}{P. Indyk},
  {and} \bibinfo{person}{R. Motwani}.} \bibinfo{year}{2012}\natexlab{}.
\newblock \showarticletitle{Approximate Nearest Neighbor: Towards Removing the
  Curse of Dimensionality.}
\newblock \bibinfo{journal}{{\em Theory of computing\/}} \bibinfo{volume}{8},
  \bibinfo{number}{1} (\bibinfo{year}{2012}), \bibinfo{pages}{321--350}.
\newblock


\bibitem[\protect\citeauthoryear{Hoeffding}{Hoeffding}{1963}]%
        {hoeffding1963}
\bibfield{author}{\bibinfo{person}{W. Hoeffding}.}
  \bibinfo{year}{1963}\natexlab{}.
\newblock \showarticletitle{Probability inequalities for sums of bounded random
  variables}.
\newblock \bibinfo{journal}{{\em Jour. Am. Stat. Assoc.\/}}
  \bibinfo{volume}{58}, \bibinfo{number}{301} (\bibinfo{year}{1963}),
  \bibinfo{pages}{13--30}.
\newblock


\bibitem[\protect\citeauthoryear{Indyk and Motwani}{Indyk and Motwani}{1998}]%
        {indyk1998}
\bibfield{author}{\bibinfo{person}{P. Indyk} {and} \bibinfo{person}{R.
  Motwani}.} \bibinfo{year}{1998}\natexlab{}.
\newblock \showarticletitle{Approximate nearest neighbors: towards removing the
  curse of dimensionality}. In \bibinfo{booktitle}{{\em Proc. {STOC} '98}}.
  \bibinfo{pages}{604--613}.
\newblock


\bibitem[\protect\citeauthoryear{Laarhoven}{Laarhoven}{2015}]%
        {laarhoven2015}
\bibfield{author}{\bibinfo{person}{T. Laarhoven}.}
  \bibinfo{year}{2015}\natexlab{}.
\newblock \showarticletitle{Tradeoffs for nearest neighbors on the sphere}.
\newblock \bibinfo{journal}{{\em CoRR\/}}  \bibinfo{volume}{abs/1511.07527}
  (\bibinfo{year}{2015}).
\newblock
\showURL{%
\url{http://arxiv.org/abs/1511.07527}}


\bibitem[\protect\citeauthoryear{Li and K{\"o}nig}{Li and K{\"o}nig}{2011}]%
        {li2011theory}
\bibfield{author}{\bibinfo{person}{Ping Li} {and}
  \bibinfo{person}{Arnd~Christian K{\"o}nig}.} \bibinfo{year}{2011}\natexlab{}.
\newblock \showarticletitle{Theory and applications of b-bit minwise hashing}.
\newblock \bibinfo{journal}{{\it Commun. ACM}} \bibinfo{volume}{54},
  \bibinfo{number}{8} (\bibinfo{year}{2011}), \bibinfo{pages}{101--109}.
\newblock


\bibitem[\protect\citeauthoryear{Mitzenmacher, Pagh, and Pham}{Mitzenmacher
  et~al\mbox{.}}{2014}]%
        {mitzenmacher2014}
\bibfield{author}{\bibinfo{person}{M. Mitzenmacher}, \bibinfo{person}{R. Pagh},
  {and} \bibinfo{person}{N. Pham}.} \bibinfo{year}{2014}\natexlab{}.
\newblock \showarticletitle{Efficient estimation for high similarities using
  odd sketches}. In \bibinfo{booktitle}{{\em Proc. {WWW} '14}}.
  \bibinfo{pages}{109--118}.
\newblock


\bibitem[\protect\citeauthoryear{Mitzenmacher and Upfal}{Mitzenmacher and
  Upfal}{2005}]%
        {mitzenmacher2005}
\bibfield{author}{\bibinfo{person}{M. Mitzenmacher} {and} \bibinfo{person}{E.
  Upfal}.} \bibinfo{year}{2005}\natexlab{}.
\newblock \bibinfo{booktitle}{{\em Probability and computing}}.
\newblock \bibinfo{publisher}{Cambridge University Press},
  \bibinfo{address}{New York, {NY}}.
\newblock


\bibitem[\protect\citeauthoryear{Motwani, Naor, and Panigrahy}{Motwani
  et~al\mbox{.}}{2007}]%
        {motwani2007}
\bibfield{author}{\bibinfo{person}{R. Motwani}, \bibinfo{person}{A. Naor},
  {and} \bibinfo{person}{R. Panigrahy}.} \bibinfo{year}{2007}\natexlab{}.
\newblock \showarticletitle{Lower Bounds on Locality Sensitive Hashing}.
\newblock \bibinfo{journal}{{\em {SIAM} J. Discrete Math.\/}}
  \bibinfo{volume}{21}, \bibinfo{number}{4} (\bibinfo{year}{2007}),
  \bibinfo{pages}{930--935}.
\newblock


\bibitem[\protect\citeauthoryear{Motwani and Raghavan}{Motwani and
  Raghavan}{2010}]%
        {motwani2010randomized}
\bibfield{author}{\bibinfo{person}{Rajeev Motwani} {and}
  \bibinfo{person}{Prabhakar Raghavan}.} \bibinfo{year}{2010}\natexlab{}.
\newblock \bibinfo{booktitle}{{\em Randomized algorithms}}.
\newblock \bibinfo{publisher}{Chapman \& Hall/CRC}.
\newblock


\bibitem[\protect\citeauthoryear{O’Donnell, Wu, and Zhou}{O’Donnell
  et~al\mbox{.}}{2014}]%
        {odonnell2014optimal}
\bibfield{author}{\bibinfo{person}{R. O’Donnell}, \bibinfo{person}{Y. Wu},
  {and} \bibinfo{person}{Y. Zhou}.} \bibinfo{year}{2014}\natexlab{}.
\newblock \showarticletitle{Optimal lower bounds for locality-sensitive hashing
  (except when q is tiny)}.
\newblock \bibinfo{journal}{{\em ACM Transactions on Computation Theory
  (TOCT)\/}} \bibinfo{volume}{6}, \bibinfo{number}{1} (\bibinfo{year}{2014}),
  \bibinfo{pages}{5}.
\newblock


\bibitem[\protect\citeauthoryear{Pagh, St{\"o}ckel, and Woodruff}{Pagh
  et~al\mbox{.}}{2014}]%
        {pagh2014min}
\bibfield{author}{\bibinfo{person}{Rasmus Pagh}, \bibinfo{person}{Morten
  St{\"o}ckel}, {and} \bibinfo{person}{David~P Woodruff}.}
  \bibinfo{year}{2014}\natexlab{}.
\newblock \showarticletitle{Is min-wise hashing optimal for summarizing set
  intersection?}. In \bibinfo{booktitle}{{\em Proceedings of the 33rd ACM
  SIGMOD-SIGACT-SIGART symposium on Principles of database systems}}. ACM,
  \bibinfo{pages}{109--120}.
\newblock


\bibitem[\protect\citeauthoryear{Panigrahy, Talwar, and Wieder}{Panigrahy
  et~al\mbox{.}}{2010}]%
        {panigrahy2010lower}
\bibfield{author}{\bibinfo{person}{R. Panigrahy}, \bibinfo{person}{K. Talwar},
  {and} \bibinfo{person}{U. Wieder}.} \bibinfo{year}{2010}\natexlab{}.
\newblock \showarticletitle{Lower Bounds on Near Neighbor Search via Metric
  Expansion}. In \bibinfo{booktitle}{{\em Proc. {FOCS} '10}}.
  \bibinfo{pages}{805--814}.
\newblock


\bibitem[\protect\citeauthoryear{Shrivastava and Li}{Shrivastava and
  Li}{2015}]%
        {shrivastava2015asymmetric}
\bibfield{author}{\bibinfo{person}{Anshumali Shrivastava} {and}
  \bibinfo{person}{Ping Li}.} \bibinfo{year}{2015}\natexlab{}.
\newblock \showarticletitle{Asymmetric minwise hashing for indexing binary
  inner products and set containment}. In \bibinfo{booktitle}{{\em Proceedings
  of the 24th International Conference on World Wide Web}}. ACM,
  \bibinfo{pages}{981--991}.
\newblock


\bibitem[\protect\citeauthoryear{Terasawa and Tanaka}{Terasawa and
  Tanaka}{2007}]%
        {terasawa2007spherical}
\bibfield{author}{\bibinfo{person}{K. Terasawa} {and} \bibinfo{person}{Y.
  Tanaka}.} \bibinfo{year}{2007}\natexlab{}.
\newblock \showarticletitle{Spherical {LSH} for Approximate Nearest Neighbor
  Search on Unit Hypersphere}. In \bibinfo{booktitle}{{\em Proc. {WADS} '07}}.
  \bibinfo{pages}{27--38}.
\newblock


\bibitem[\protect\citeauthoryear{Thorup}{Thorup}{2013}]%
        {thorup2013bottom}
\bibfield{author}{\bibinfo{person}{Mikkel Thorup}.}
  \bibinfo{year}{2013}\natexlab{}.
\newblock \showarticletitle{Bottom-k and priority sampling, set similarity and
  subset sums with minimal independence}. In \bibinfo{booktitle}{{\em
  Proceedings of the forty-fifth annual ACM symposium on Theory of computing}}.
  ACM, \bibinfo{pages}{371--380}.
\newblock


\bibitem[\protect\citeauthoryear{Tops\oe}{Tops\oe}{2007}]%
        {topsoe2007}
\bibfield{author}{\bibinfo{person}{F. Tops\oe}.}
  \bibinfo{year}{2007}\natexlab{}.
\newblock \bibinfo{booktitle}{{\em Some Bounds for the Logarithmic Function}}.
  Vol.~\bibinfo{volume}{4}.
\newblock \bibinfo{publisher}{Nova Science}, \bibinfo{pages}{137--151}.
\newblock


\bibitem[\protect\citeauthoryear{Zobrist}{Zobrist}{1970}]%
        {zobrist1970new}
\bibfield{author}{\bibinfo{person}{Albert~L Zobrist}.}
  \bibinfo{year}{1970}\natexlab{}.
\newblock \showarticletitle{A new hashing method with application for game
  playing}.
\newblock \bibinfo{journal}{{\em ICCA journal\/}} \bibinfo{volume}{13},
  \bibinfo{number}{2} (\bibinfo{year}{1970}), \bibinfo{pages}{69--73}.
\newblock


\end{thebibliography}
\end{document}